\documentclass[peerreview,onecolumn]{IEEEtran}
\usepackage{stackengine}
  \usepackage{tikz}
\usetikzlibrary{arrows,%
                shapes,positioning}

\usetikzlibrary{arrows,%
                petri,%
                topaths}%
\usepackage{tkz-berge}

\usepackage{array,tabularx, float}

\usepackage{babel}
\usepackage{verbatim}
\usepackage{psfrag}
\usepackage{graphicx}
\usepackage{subfigure}
\usepackage{amsmath,amssymb,lineno,enumerate,amsthm,amsfonts}

\newcommand{\FF}{\mathbb F}

\def\cC{\mathcal C}

\def\cF{\mathcal F}

\def\cI{\mathcal I}

\def\cJ{\mathcal J}

\def\cL{\mathcal L}
\def\cM{\mathcal M}
\def\cN{\mathcal N}

\def\cS{\mathcal S}

\def\cW{\mathcal W}
\def\cX{\mathcal X}
\def\cY{\mathcal Y}
\def\cZ{\mathcal Z}

\def\min{{\rm min}}

\def\dim{\mbox{\rm dim}}

\newcommand{\ga}{\alpha}
\newcommand{\gb}{\beta}

\newcommand{\gd}{\delta}
\newcommand{\gep}{\varepsilon}

\newcommand{\gl}{\lambda}
\newcommand{\gk}{\kappa}

\newcommand{\rk}{{\rm{rank}}}

\newcommand{\fq}{{\mathbb F}_q}
\newcommand{\rank}{{\mathrm{rk}}}
\newcommand{\ho}{{\mathrm{Hom}}}
\newcommand\qbin[3]{\left[\begin{matrix} #1 \\ #2 \end{matrix} \right]_{#3}}

\newcommand{\be}{{\bf e}}

\newtheorem{theorem}{Theorem}[section]
\newtheorem{proposition}[theorem]{Proposition}
\newtheorem{lemma}[theorem]{Lemma}
\newtheorem{corollary}[theorem]{Corollary}

\newtheorem{remark}[theorem]{Remark}

{\theoremstyle{definition}
\newtheorem{definition}[theorem]{Definition}
\newtheorem{example}[theorem]{Example}

\newtheorem*{proposition*}{Proposition}
\newtheorem*{corollary*}{Corollary}
\newtheorem*{lemma*}{Lemma}

\usepackage{color}

\def\INSe#1{{\color{black}#1}}
\def\INSm#1{{\color{black}#1}}
\ifCLASSINFOpdf
\else
\fi
\hyphenation{op-tical net-works semi-conduc-tor}

\begin{document}
\title{Error Correction for Index Coding With Coded Side Information}

\author{Eimear~Byrne,
        and~Marco~Calderini. 
\thanks{School of Mathematical Sciences, 
University College Dublin, Ireland.}
\thanks{e-mail: ebyrne@ucd.ie}
\thanks{Research supported by ESF COST Action IC1104}% <-this % stops a space
\thanks{Department of Mathematics,
 University of Trento, Italy.}
\thanks{email: marco.calderini@unitn.it}
\thanks{Research supported by ESF COST Action IC1104}
\thanks{Manuscript received MONTH, YEAR.}
% \thanks{This paper was presented at ???.}
\thanks{Copyright (c) 2013 IEEE. Personal use of this material is permitted. However, permission to use this material for any other purposes must be obtained from the IEEE by sending a request to pubs-permissions@ieee.org.}
}

\markboth{Journal of \LaTeX\ Class Files,~Vol.~??, No.~?, Month~YEAR}%
{Shell \MakeLowercase{\textit{et al.}}: Bare Demo of IEEEtran.cls for Journals}

% Remember, if you use this you must call 0,00\IEEEpubidadjcol in the second


\maketitle

\begin{abstract}
  \INSe{Index coding is a source coding problem in which a broadcaster seeks to meet the different demands of several users, each of whom is assumed to have some prior information on the data held by the sender. A well-known application is to satellite communications, as described in one of the earliest papers on the subject \cite{BK}. It is readily seen that if the sender has knowledge of its clients' requests and their {\em side-information} sets, then the number of packet transmissions required to satisfy all users' demands can be greatly reduced if the data is encoded before sending. The collection of side-information indices as well as the indices of the requested data is described as an {\em instance} $\cI$ of the index coding with side-information (ICSI) problem. The encoding function is called the index code of $\cI$, and the number of transmissions employed by the code is referred to as its {\em length}. The main ICSI problem is to determine the optimal length of an index code for and instance $\cI$. As this number is hard to compute, bounds approximating it are sought, as are algorithms to compute efficient index codes. These questions have been addressed by several authors \cite{alon,BBJK,BBJK11,BKL,SDL,tehrani2012bipartite}, often taking a graph-theoretic approach. Two interesting generalizations of the problem that have appeared in the literature are the subject of this work. The first of these is the case of index coding with {\em coded side information} \cite{dai,shum}, in which linear combinations of the source data are both requested by and held as users' side-information. This generalization has applications, for example, to relay channels and necessitates algebraic rather than combinatorial methods. The second is the introduction of error-correction in the problem, in which the broadcast channel is subject to noise \cite{Son1}.
  In this paper we characterize the optimal length of a scalar or vector linear index code with coded side information (ICCSI) over a finite field in terms of a generalized {\em min-rank} and give bounds on this number based on constructions of random codes for an arbitrary instance. We furthermore consider the length of an optimal $\delta$-error correcting code for an instance of the ICCSI problem and obtain bounds analogous to those described in \cite{Son1}, both for the Hamming metric and for rank-metric errors. We describe decoding algorithms for both categories of errors based on those given in \cite{Son1,sil}}.

\end{abstract}

% \fi
\IEEEpeerreviewmaketitle
%
%%%%%%%%%%%%%%%%%%%%%%%%%%%%%%%%%%%%%%%%%%%%%%%%%%%%%%%%%%%%%%%%%%%%%%%%%%
%

\section{Introduction}
\label{intro}

\INSe{The problem of index coding with side information (ICSI) was introduced by Birk and Kol in \cite{BK} under the term {\em informed source coding on demand}. In \cite{BBJK} the authors explicitly refer to the problem as {\em index coding}. This topic is motivated by applications in broadcast communications such as audio and video on-demand, content delivery, and wireless networking. It relates to a problem of source coding with {\em side information}, in which receivers have partial information about the data to be sent prior to its broadcast. The problem for the sender is to exploit knowledge of the users' side information to encode data optimally, that is to reduce the overall length of the encoding, or equivalently, the number of transmitted packets. The ICSI problem has since become a subject of several studies and generalizations \cite{alon,BBJK,BBJK11,RSG,Son1,Son,SDL}.} 

The scenario of the ICSI problem is the following. A server (sender) has to broadcast some data to a set of clients (receivers or users), with possibly different messages requested by different clients. \INSe{Before the transmission starts, each receiver already has some data in its possession, its cached packets, called its side-information. These packets may be from a previous broadcast, perhaps sent during lighter data traffic periods, or acquired by some other communication. The receivers let the sender know which messages they have, and which they require. The broadcaster can use this information, along with encoding, to reduce the overall number of packet transmissions required to satisfy all the demands of its clients. If the sender has been successful in this endeavour, then the broadcasted data can be utilized by each user, along with its cached packets, in order to decode its own specific demand.} 

The main index coding problem is to determine the minimum number of packet transmissions required by the sender in order to satisfy all users' requests, if encoding of data is permitted. Given an instance of the ICSI problem, Bar-Yossef {\emph{ et al}} \cite{BBJK} proved that finding the best {\em scalar linear} binary index code is equivalent to finding the {\em min-rank} of a graph, which is known to be an NP-hard problem \cite{Pee}. \INSe{The twin problem is to determine an explicit optimal encoding function for an instance. Any encoding function for an instance necessarily gives an upper bound on the optimal length of an index code. There have been a number of papers addressing this aspect of the problem, in fact finding sub-optimal but feasible solutions, using linear programming methods to obtain partitions of the users into solvable subsets. Such solutions involve obtaining {\em clique covers, partial-clique covers, multicast partitions} and some variants of these \cite{BKL,BC,
		SDL,shum,tehrani2012bipartite}. Other than these LP approaches, low-rank matrix completion methods may also be applied. This was considered for index coding over the real numbers in \cite{HeR}.}

\INSe{The importance of the index coding problem can also be seen in its equivalences and connections to other problems, such as {\em network coding}, {\em coded-caching} and {\em interference alignment} \cite{EeRL,RSG,caching1,MCJ}. These equivalences mean that results in index coding have impact in such other areas, and vice versa.} 

\INSe{In \cite{shum,dai} the authors give a generalization of the index coding problem in which both demanded packets and locally cached packets may be linear combinations of some set of data packets. We refer to this as the {\em index coding with coded side information} problem (ICCSI). This represents a significant departure from the ICSI problem in that an ICCSI instance no longer has an obvious association to a graph, digraph or hypergraph, as in the ICSI case. However, as we show here, it turns out that many of the results for index coding have natural extensions in the ICCSI problem.}

\INSm{One motivation for the ICCSI generalization is related to the coded-caching problem. The method in \cite{caching1} uses uncoded cache placement, but the authors give an example to show that coded cache placement performs better in general. In \cite{caching2}, it is shown that in a small cache size regime, when the number of users is not less than the number of files, a scheme based on coded cache placement is optimal. Moreover in \cite{caching3} the authors show that the only way to improve the scheme given in \cite{caching1} is by coded cache placement.}

\INSe{Another motivation is toward applications for wireless networks with relay helper nodes and cloud storage systems (see \cite{dai} and the references therein).}
\INSm{
Consider the example in Table \ref{tab:ex}. We have a scenario with one sender and four receivers $U_1$, $U_2$, $U_3$ and $U_4$. The source node has four packets $X_1, X_2, X_3$ and $X_4$ and for $i=1,...,4$ user $U_i$ wants packet $X_i$. The transmitted packet is subject to independent erasures. It is assumed that there are feedback channels from the users, informing the transmitting node which packets are successfully received. 
At the beginning, in time slot 1, 2, 3 and 4 the source node transmits packets $X_1$, $X_2$, $X_3$ and $X_4$, respectively. After time slot 4 we have the following setting: $U_1$ has packet $X_2$, $U_2$ has packet $X_1$, $U_3$ has packet $X_4$ and  $U_4$ has packet $X_4$. Now from the classical ICSI problem we have that receivers $U_1$ and $U_2$ form a clique, in the associated graph, and then we can satisfy their request sending $X_1+X_2$. Similarly for $U_3$ and $U_4$ we can use $X_3+X_4$. So, the source node in time slot 5 and 6 transmits the coded packet $X_1 + X_2$ and $X_3+X_4$, intending that users receive the respective packet.
However, $U_1$ and $U_2$ receive the coded packet $X_3+X_4$ and $U_3$ and $U_4$ receive $X_1+X_2$.
 At this point if only the uncoded packets in their caches are used, we still need to send two packets. If all packets in their caches are used, the source only needs to transmit one coded packet $X_1 + X_2 + X_3+X_4$ in time slot 7. If all four users can receive this last transmission successfully, then all users can decode the required packets by linearly combining with the packets received earlier.
}
\begin{table}[h!]
\centering
\begin{tabular}{| c |c|c|c|c|c|}
\hline
\stackanchor{Time}{slot} & Packet sent & \stackanchor{Received}{by $U_1$?}& \stackanchor{Received}{by $U_2$?}&\stackanchor{Received}{by $U_3$?}&\stackanchor{Received}{by $U_4$?}\\
\hline
1&$X_1$&no&yes&no&no\\
\hline
2&$X_2$&yes&no&no&no\\
\hline
3&$X_3$&no&no&no&yes\\
\hline
4&$X_4$&no& no&yes&no\\
\hline
5&$X_1+X_2$&no&no &yes&yes\\
\hline
6&$X_4+X_3$&yes&yes &no&no\\
\hline
7&$X_1+X_2+X_3+X_4$&yes&yes &yes&yes\\
\hline

\end{tabular}\label{tab:ex}
\caption{Illustration of utilizing coded packets as side information.}
\end{table}

\begin{figure}[h!]
\centering

\begin{tikzpicture}[->,>=stealth',shorten >=0.2pt,auto,node distance=1.7cm,
  thick,main node/.style={font=\sffamily\bfseries\small},scale=0.61]

  \node[main node] (1)  [label=left:\small$X_1+X_2+X_3+X_4$]{\includegraphics[scale=0.12]{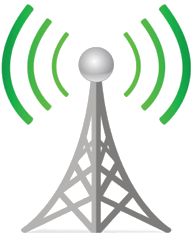}};
  \node[main node] (2) [below left of=1]  [label=below left:\mbox{\small has $X_2$, $X_3+X_4$}] [label=left:\small\mbox{\small\color{blue} wants $X_1$}]{\includegraphics[scale=0.4]{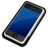}$U_1$};
  
  \node[main node] (3) [below right of=1]  [label=below right:\small\mbox{ \small has $X_1$, $X_3+X_4$}] [label=right:\small\mbox{\small \color{blue} wants $X_2$}]{\includegraphics[scale=0.4]{PhoneIcon.png}$U_2$};
   \node[main node] (4) [above right of=1]  [label=above right:\small\mbox{\small has $X_4$, $X_1+X_2$}] [label=right:\small\mbox{\small \color{blue} wants $X_3$}]{\includegraphics[scale=0.4]{PhoneIcon.png}$U_3$};
      \node[main node] (5) [above left of=1]  [label=above left:\small\mbox{\small has $X_3$, $X_1+X_2$}] [label=left:\small\mbox{\small \color{blue} wants $X_4$}]{\includegraphics[scale=0.4]{PhoneIcon.png}$U_4$};
%\draw[dashed, black] (0, 0) circle [radius=1.4]; % C
%	\draw[dashed, black] (0, 0) circle [radius=2.4]; % D

\end{tikzpicture}\label{fig:ex}
\caption{State of network after time slot 6.}
\end{figure}

\INSe{A second generalization of the ICSI problem was given in \cite{Son1}, where the authors consider error correction. That is, the broadcast channel may be subject to noise during a transmission. Classical coding theory plays a role in several of the results and a number of bounds are given on the optimal length of an {\em error correcting index code} (ECIC) that corrects some $\gd$ Hamming errors. A decoding algorithm based on syndrome decoding is also described. We remark that error-correction for network coding has only been addressed for multicast, where rank-metric and subspace codes are proposed. There are numerous papers on this subject after the seminal works \cite{KK08,SKK08}. Many of these are based on Gabidulin codes \cite{G85}.}

\INSe{\subsection{Our Contribution}}

\INSe{In this paper we develop the theory of index coding further along the lines of these latter mentioned generalizations. That is, we consider the ICCSI problem both in the error-free case and with respect to error correction. We assume that the source data is composed of $n$ blocks of length $t$ over $\fq$ (the finite field of $q$ elements),
and that encoding involves taking $\fq$-linear combinations of the $n$ data blocks.
In particular, we consider both linear and scalar-linear index codes.
We describe a generalized min-rank for the ICCSI problem, which we show gives the optimal length for $\fq$-linear encodings. This quantity is actually shown to be the minimum rank weight of the coset of an $\fq$-linear matrix code determined by an ICCSI instance.  
We characterize necessary and sufficient conditions for a matrix $L$ to {\em realize} an instance of the ICCSI problem and use this to obtain upper bounds on the length of an optimal $\fq$-linear index code with coded side information. The first of these may be viewed as a generalization of the bound obtained by the existence of a partial clique in the {\em side-information graph} of a classical index coding problem. It requires $q$ to be large although it does not rely on the use of a {\em maximum distance separable} (MDS) code. The second of these bounds offers a refinement and relaxation of the constraint on $q$ and is not explicit. Both are based on the probability that an arbitrary matrix realizes a code for an instance.}

Following the work of \cite{Son1}, we consider error correction for the ICCSI problem, both for the Hamming and rank metric and address the question of the main index coding problem for error correcting index codes. \INSe{We establish criteria for error correction for an ICCSI instance and give bounds on the optimal length of a $\delta$-error correcting ECIC, both for the Hamming metric and the rank metric. These results are extensions of the $\gk,\ga$ and sphere-packing and Singleton bounds as described in \cite{Son1}. Some of these also yield further upper bounds on the optimal length of an ICCSI code for the error-free case.}

\INSe{Finally, we outline decoding strategies for linear ECICs for both the rank and Hamming distance. In the first case we extend the syndrome decoding method to correct Hamming errors for index codes given in \cite{Son1} to the ICCSI case. In the second, we show that the simple, low-complexity strategy for additive matrix channels given in \cite{sil} can be applied to correct rank-metric errors, that is to handle error matrices of rank upper bounded by some $\gd$.}

\section{Preliminaries}\label{sec:pre}
%\subsection{Notation}

We establish notation to be used throughout the paper.
For any positive integer $n$, we let $[n]:=\{1,\dots,n\}$. We write $\fq$ to denote the finite field of order $q$ and use $\fq^{n\times t}$ to denote the vector space of all $n\times t$ matrices over $\fq$. 
Given a matrix $X \in \fq^{n\times t}$ we write ${ X}_i$ and ${ X}^j$ to denote the $i$th row and $j$th column of $X$, respectively. More generally, for subsets ${\mathcal S}\subset [n]$ and ${\mathcal {\INSe{T}}}\subset [t]$ we write ${X}_{\mathcal S}$ and 
${ X}^{\mathcal T}$ to denote the $|{\mathcal S}| \times t$ and $n \times |{\mathcal T}|$ submatrices of $X$ comprised of the rows of $X$ indexed by ${\mathcal S}$ and the columns of $X$ indexed by ${\mathcal T}$ respectively.  We write $\langle X \rangle$ to denote the row space of $X$.

In this work we will consider two distance functions, namely the Hamming metric and the rank metric, over the $\fq$-vector space
$\fq^{n\times t}$.

\INSe{Choosing a basis of the finite field of $q^t$ elements, it is easy to see that $\FF_{q^t}$ and $\fq^{n \times t}$ are isomorphic as $\fq$-vector spaces. Then, given the usual definition of the Hamming distance between a pair of elements $x,y \in \FF_{q^t}^n$:
	$$d_H(x,y):=|\{ i : x_i \neq y_i \} |,$$
	we define the Hamming distance between a pair of matrices $X,Y \in \fq^{n \times t}$ as the number of coordinates in $[n]$ such that
	$X_i \neq Y_i$, so the number of differing rows of $X$ and $Y$.}

For two matrices $A,B \in \fq^{n\times t}$, the rank distance between $A$ and $B$ is the rank of the matrix $A-B$ over $\fq$:
$$d_{\rm{rk}}(A,B)= \rank(A-B). $$
We write $d(A,B)$ to denote either distance function between $A$ and $B$ and     
we write $w(A)$ to denote $d(A,0)$. Given $A$ and a set $\cS$, $d(A,\cS) = \min \{ d(\INSe{A},S) : S \in \cS\}$.
\INSe{In some cases we will specify explicitly which distance function should be understood, otherwise the reader should interpret $d$ or $w$ as denoting either metric.}

\INSe{Recall that for any pair of subspaces $U$ and $V$, their sum is the subspace $U+V = \{u + v : u \in U, v \in V\}$ and we write $U \oplus V$ to denote the direct sum
$ U \oplus V = \{(u,v) : u \in U, v \in V\}$. Moreover $U+V$ and $U \oplus V$ are isomorphic if and only if $U \cap V$ is the trivial space. For arbitrary $x$ in the ambient space, the coset $x+U:=\{ x+u : u \in U \}$.}
\INSe{We use the standard notation $U < V$ to denote that $U$ is a subspace of $V$.}

\section{Index coding with coded side information}\label{sec:iccsi}

%In this section we consider a model in which a receiver can have coded packets as its side information, as discussed in \cite{shum}. There the 
In \cite{shum} the authors generalized the index coding problem so that coded packets of a data matrix $X$ may be broadcast or part of a user's cache. 
As mentioned before, this finds applications, in broadcast channels with helper relay nodes. 

\INSe{Before we present the model with coded side information, let us recall the scenario for uncoded side information (see \cite{Son1,Son}).
In that case, the data is a vector $X \in \fq^n$ possessed by a single sender. There are $m$ users or receivers, each of which has an index set $\cX_i \subset [n]$, called its side-information. This indicates that the $i$th user possesses the entries of $X$ indexed by $\cX_i$. The surjection $f:[m] \longrightarrow [n]$ assigns users to indices, indicating that User $i$ wants $X_{f(i)}$ and it is also assumed that $f(i) \notin \cX_i$. The sender is assumed to be informed of the values $f(i)$ and $\cX_i$ of each user.}

We now describe an instance of index coding with coded-side information.
There is a data matrix $X \in \FF_q^{n \times t}$ and a set of $m$ receivers or users. $X$ is thus a list of $n$ blocks of length $t$ over $\fq$.
%So in this scenario we have the message ${\bf x}=(x_1,\dots,x_n)$. 
For each $i \in[m]$, the $i$th user seeks some linear combination of the rows of $X$, say $R_i X$ for some $R_i \in \FF_q^n$. 
\INSe{We'll refer to $R_i$ as the request vector and to $R_iX$ as the request packet of User $i$.}
A user's cache \INSe{denotes locally stored data, which it can freely access. In our model it} is represented by a pair of matrices 
$$V^{(i)} \in \FF_q^{d_i \times n} \text{ and }\Lambda^{(i)}\in \FF_q^{d_i \times t}$$
related by the equation
$$\Lambda^{(i)} = V^{(i)}X.$$ 
While the matrix $X$ may be unknown to User $i$, it is assumed that any vector in the row spaces of $V^{(i)}$ and $\Lambda^{(i)}$ can be generated at the $i$th receiver.
We denote these respective row spaces by $\cX^{(i)}:=\langle V^{(i)} \rangle  $ and ${\mathcal L}^{(i)}:=\langle \Lambda^{(i)} \rangle$ for each $i$.
The side information of the $i$th user is $(\cX^{(i)},{\mathcal L}^{(i)})$.
Similarly, the sender \INSm{$S$} has the pair of row spaces $(\cX^{(S)}, {\mathcal L}^{(S)})$ for matrices 
$$V^{(S)} \in \FF_q^{d_S \times n} \text{ and } \Lambda^{(S)} = V^{(S)}X \in \FF_q^{d_S \times t}$$ 
and does not necessarily possess the matrix $X$ itself.

The $i$th user requests a coded packet $R_iX \in {\mathcal L}^{(S)}$ with $R_i \in \cX^{(S)} \backslash \cX^{(i)}$. 
We denote by $R$ the $m \times n$ matrix over $\fq$ with each $i$th row equal to $R_i$. The matrix $R$ thus represents the requests of all $m$ users.
%Let $d=\sum_{i \in [m]} d_i$. We let ${\mathcal V}$ denote the $d \times n$ matrix over $\fq$ satisfying:
%\begin{equation}
%\label{eqV} V := \left[ \begin{array}{c} V^{(1)} \\ V^{(2)} \\ \vdots \\ V^{(m)} \end{array} \right].
%\end{equation}
We denote by 
$$ \cX: = \{ A \in \fq^{m \times n} : A_i \in \cX^{(i)}, i \in [m]\},$$
so that $\cX=\oplus_{i \in [m]} \cX^{(i)}$ %and ${\mathcal L}:=\oplus_{i \in [m]}{\mathcal L}^{(i)}$ 
is the direct sum of the $\cX^{(i)} $ %and the ${\mathcal L}^{(i)}$ respectively 
as a vector space over $\fq$. 

\noindent\INSe{We define $\tilde{\cX}  := \{Z \in \fq^{m \times n} : Z_i \in \cX^{(S)}\}$, which may be viewed as the direct sum of $m$ copies of $ \cX^{(S)}$}. 

\begin{remark}
	The reader will observe that the classical ICSI problem is indeed a special case of the index coding problem with coded side information (cf. \cite{Son1,Son}). Setting $V^{(S)}$ to be the $n \times n$ identity matrix, $R_i = {\bf e}_{f(i)} \in \fq^n$ %\INSe{(where the assignment of packets to users is $f:[m]\longrightarrow [n]$)} 
	and $V^{(i)}$ to be the $d_i \times n$ matrix with rows $V^{(i)}_j = {\bf e}_{i_j}$ for each $i_j \in \cX_i$, yields $\cX^{(i)} = \langle {\bf e}_{j} : j \in \cX_i \rangle$. \INSe{Then User $i$ has the rows of $X$ indexed by $\cX_i$ and requests $X_{f(i)}$.} %, so that $\supp (v) \subset \cX_i$ if and only if $v \in \cX^{(i)}$.  
\end{remark}

\begin{remark}
\INSm{The case where the sender does not necessarily possess the matrix $X$ itself can be applied to the {\em broadcast relay channel}, as described in \cite{shum}. The authors consider a channel as in Fig. \ref{fig:relay}, and assume that the relay is close to the users and far away from the source, and in particular that all relay-user links are erasure-free.
\INSe{Each node is assumed to have some storage capacity and stores previously received data in its cache.}	
The packets in the cache of the relay node \INSe{are obtained as previous broadcasts, hence it may contain both} coded and uncoded packets. The relay node, playing the role of the sender, transmits packets obtained by linearly combining the packets in its cache, depending on the requests and coded side information of all users. It seeks to mimimize the total number of broadcasts such that all users' demands are met.}
\end{remark}

\begin{figure}
\centering
\includegraphics[scale=0.15]{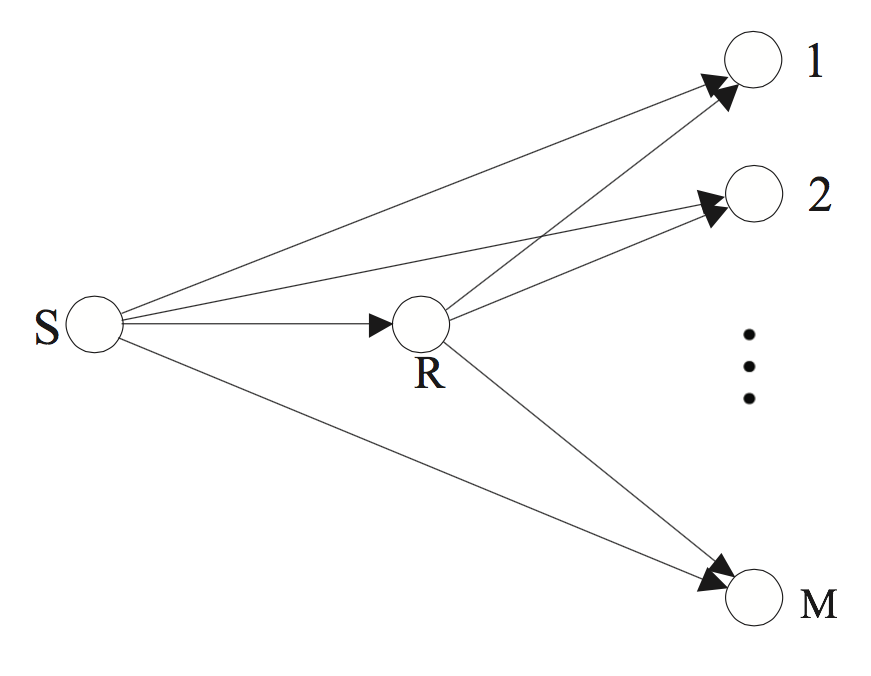}
\caption{Broadcast Relay Channel}\label{fig:relay}
\end{figure}

\begin{definition}

An instance of the Index Coding with Coded Side Information (ICCSI) problem is a list  $\cI=(t,m,n,\cX,\cX^{(S)},R)$
for some positive integers $t,m,n$, subspaces $\cX^{(S)}$ and $\cX^{(i)}$ of $\fq^{n}$ of dimensions $d_S,d_i$ for $i \in [m]$
such that $\cX = \oplus_{i \in [m]} \cX^{(i)}$ and a matrix $R$ in $\tilde{\cX}$.
\end{definition} 

For the remainder, we let $t,m,n,\cX,\cX^{(S)},\tilde{\cX} ,R$ be as described above and we fix $\cI=(t,m,n,\cX,\cX^{(S)},R)$ to denote an instance of the ICCSI problem for these parameters. \INSe{We now define what is meant by an index code for an instance $\cI$: it is essentially a map that encodes any data matrix $X$ in such a way that each user, given its side-information and received transmission, can uniquely determine its requested packet $R_iX\in \fq^t$.}

\begin{definition}
Let $N$ be a positive integer. We say that the map	
$$
E:\fq^{n\times t}\to\fq^{N \times t},
$$
is an $\fq$-code for $\cI$ of length $N$ if for each $i$th receiver, $i \in [m]$ there exists a decoding map

$$
D_i:\fq^{N\times t}\times\cX^{(i)}\to\fq^t,
$$
satisfying
$$
\forall X\in\fq^{n \times t} \,:\, D_i(E({X}),A)=R_i X,
$$
for some vector $A \in \cX^{(i)}$, in which case we say that $E$ is an $\cI$-IC. 
$E$ is called an $\fq$-linear $\cI$-IC if $E(X)=LV^{(S)}X$ for some $L\in\fq^{N \times d_S}$, in which case we say that $L$ represents
the code $E$, \INSe{or that the matrix $L$ realizes $E$}. \INSe{If $t=1$, we say that $L$ represents a scalar linear index code. If $t>1$ we say that the code is vector linear. We write $\cL$ to denote the space $\langle L V^{(S)} \rangle.$ 
	}
\end{definition}

\INSe{An encoding is sought such that the length $N$ of the $\cI$-IC is as small as possible. We shall be principally concerned with $\fq$-linear codes for an instance $\cI$.}
\INSe{We assume that the side information matrices $V^{(i)}$ of all users are known to the sender, along with the demand vectors $R_i$. As we'll see in the next section, this knowledge is sufficient to determine an encoding matrix $L$ for an $\fq$-linear $\cI$-IC. These assumptions are in keeping with those outlined in \cite{BK} for the original {\em informed source coding on demand problem} and are based on the existence of a slow error-free reverse channel allowing communication from users to the sender. We also assume that $LV^{(S)}$ is known to the receivers before the broadcast of the encoded matrix $LV^{(S)}X$. This knowledge, along with the transmission $LV^{(S)}X$ and its own cache data will be used by each $i$th user in order to compute its demand $R_iX$. These assumptions mean that the gains of encoding an ICCSI instance are greater as $t$ increases.} %Clearly, all users' requests can be met by simply transmitting $X$, which could cost up to $n$ transmissions of $t$-blocks over $\fq$. The aim of the encoding is to reduce the total number of block transmissions, so the encoded matrix $E(X)$ should have much fewer rows than the original matrix $X$. 

\INSe{\subsection{Necessary and Sufficient Conditions for Realization of an $\fq$-Linear $\cI$-IC} }

\INSe{In the following we give necessary and sufficient conditions for a matrix $L$ to represent a linear code of the instance $\cI$ (in fact the sufficiency of the statement of Lemma \ref{lemdecode} has already been noted in \cite{shum}). }

\begin{lemma}\label{lemdecode}
   Let $L\in \fq^{N \times d_S}$. Then $L$ represents an $\fq$-linear $\cI$-IC index code of length $N$ if and only if for each $i \in [m]$,
   $R_i \in \cL + \cX^{(i)}. $
   %$$ R_i \in \left\langle \left[ \begin{array}{r} V^{(i)} \\ LV^{(S)} \end{array} \right]\right\rangle.$$
\end{lemma}

\begin{proof}
   Let $i\in [m]$ and let $R_i \in \cX^{(S)}$. Suppose that $Y = LV^{(S)}X$ has been transmitted.
   If $R_i \in \cL + \cX^{(i)} $
   %$$R_i \in \left\langle \left[ \begin{array}{r} V^{(i)} \\ LV^{(S)} \end{array} \right]\right\rangle,$$
   then there exist $A \in \fq^{d_i},B \in \fq^N$ such that $R_i = AV^{(i)} + B LV^{(S)}$. 
   Then for any $X \in \fq^{n \times t}$ we have
   $$R_iX = AV^{(i)}X +B LV^{(S)}X = A\Lambda^{(i)} + BY. $$
   Therefore, \INSe{Receiver $i$, knowing $V^{(i)}X$ and $\Lambda^{(i)}$, can compute $A$ and $B$ and hence} acquires $R_iX$.
   
   Conversely, suppose that $R_i \notin \cL + \cX^{(i)}. $
   %$R_i$ is not contained in the row space of the matrix $\left[ \begin{array}{r} V^{(i)} \\ LV^{(S)} \end{array} \right]$ for some $i \in [m]$.
   Then for each $U \in \fq^t$, we have %there exist $|{\bf u}^\perp|^t$ elements $X \in \fq^{n \times t}$ satisfying $R_iX = {\bf u}$.
   \begin{eqnarray*}
   	\rk\left( \left[ \begin{array}{ll}
   		R_i & U \\
   		V^{(i)} & \Lambda^{(i)}\\
   		LV^{(S)} & Y
   	\end{array}    
   	\right]\right) 
   	& = & 1 + \rk\left( \left[ \begin{array}{ll}
   		V^{(i)} & \Lambda^{(i)}\\
   		LV^{(S)} & Y
   	\end{array}    
   	\right]\right)\\
   	& = & 1 + \rk\left( \left[ \begin{array}{ll}
   		V^{(i)}\\
   		LV^{(S)} 
   	\end{array}    
   	\right]\right)
   	=
   	\rk\left(\left[ \begin{array}{l}
   		R_i \\
   		V^{(i)}\\
   		LV^{(S)}
   	\end{array}    
   	\right]\right).
   \end{eqnarray*} 
   In particular, the linear system
    $$R_iX = U , V^{(i)}X =\Lambda^{(i)}, LV^{(S)}X = Y $$
    is consistent for each $\INSe{U} \in {\fq}^t$. 
    It follows that
    \begin{gather}\label{eq:pr}
       Pr(R_iX = \INSe{U} | V^{(i)}X =\Lambda^{(i)},  LV^{(S)}X = Y)=\frac {1}{q^t},
    \end{gather}                  
    so the side information of $R_i$ conveys no information about $R_iX$ to \INSe{the $i$th receiver}.
\end{proof}
%
%We remark that the sufficiency of the \INSe{statement} of Lemma \ref{lemdecode} has already been noted in \cite{shum}.

\INSe{Lemma \ref{lemdecode} simply says that the demands of all users can be simultaneously satisfied if and only if for each $i$ the smallest vector space containing both $\cL$ and $\cX^{(i)}$ also contains $R_i$; in other words extending the side-information spaces $\cX^{(i)}$ by the same space $\cL$ in each case
	  contains the $i$th request vector. This is achieved, for example, if $\cL + \cX^{(i)}$ is the space $\cX^{(S)}$ for each $i$, although 
	  this is clearly not necessary.}  %This observation is the basis for the proof of Theorem \ref{th:random1}, which gives an upper bound on the  
	  %optimal length of a linear index code for an instance for sufficiently large $q$.}
	
	\INSe{An equivalent formulation of the statement of Lemma \ref{lemdecode} is to say that $L$ represents a linear index code for $\cI$ if and only if $\cL$ meets each coset $R_i +(\cX^{(i)}\cap \cX^{(S)}) = \{R_i + A : A \in \cX^{(i)}\cap \cX^{(S)}  \}$.
	We will use this view to obtain an upper bound on the optimal length on a linear index code in Theorem \ref{thbd2}. }

Given an $ \ell \times n$ matrix $A \in \FF_q^{\ell \times n}$, we write $A^\perp$ to denote the null space of $A$ in $\fq^n$. %, that is:
%%$$A^\perp:=\{ Z \in \F_q^{n \times t} : AZ = 0\}. $$
%$$A^\perp:=\{ Z \in \F_q^{n} : A Z = 0\}. $$
Furthermore, for each $i \in [m]$ we define the sets:
\begin{eqnarray*}
	\cY^{(i)} &:=&\{Z \in \fq^{n \times t} : V^{(i)}Z = 0\}=({V^{(i)}}^\perp)^t, \\
    %\cW^{(i)} &:=&\{Z \in \fq^{n \times t} : V^{(i)}Z = 0, R_i Z = 0 \} = ({V^{(i)}}^\perp \cap {R_i}^\perp)^t \\
    \cZ^{(i)}  &:=&\{Z \in \fq^{n \times t} : V^{(i)}Z = 0, R_i Z \neq 0 \}.% = \cY^{(i)} \backslash \cW^{(i)}
\end{eqnarray*}
\INSe{To help put these sets in context, if $V^{(i)}$ has rows composed of standard basis vectors, say with leading ones indexed by the set $\cS^i \subset [n]$ (which means the side-information of User $i$ is uncoded) then $\cY^{(i)}$ consists of those matrices whose columns indexed by $\cS^i$ are all-zero. Then $\cY^{(i)}$ can be identified with the set $[n]\backslash \cX_i$, the complement of the side information of user $i$ and $\cZ^{(i)}$ can be identified with $[n]\backslash \cX_i \cup \{f(i) \}$. }
\begin{remark}
\INSm{In the classical ICSI problem, two \INSe{data matrices} $X$ and $X'$ are \INSe{called} \emph{confusable at receiver $i$} (cf. \cite{AK}) if they yield the same side information for $i$, i.e. $X_j=X'_j$ for all $j\in\cX_i$, and \INSe{if moreover} the packets $X_{f(i)}$ and  $X'_{f(i)}$ are different (here $\cX_i$ represents the side information of the receiver $i$ and $f(i)$ the request packet). In the ICCSI problem, two \INSe{vectors $X,X'$ are called confusable at receiver $i$ if $V^{(i)}X=V^{(i)}X'$ and $R_iX\ne R_iX'$, i.e. if they yield the same side information for the $i$th user but the requested data packets are different. Therefore,} $X$ and $X'$ are confusable at receiver $i$ if and only if $X-X'$ lies in the set $\cZ^{(i)}$.} 
	
	\INSe{The essential content of next result, which follows from Lemma \ref{lemdecode}, is that
	$L$ represents a linear code of $\cI$ if and only if any confusable pair $X,X'$ result in different encodings. 
	Therefore, another way of stating Corollary \ref{cordec} is:\\
	\begin{center} 
	$L$ represents an $\fq$-linear $\cI$-IC if and only if $LV^{(S)}X \neq LV^{(S)}X'$ for any confusable pair $X,X' \in \fq^{n \times t}$. \\
	\end{center}}
	
\end{remark}

\INSe{Then $L$ realizes and $\fq$-linear $\cI$-IC if and only if no matrix of $\cZ^{(i)}$ vanishes after multiplication by $LV^{(S)}$, so $\cZ^{(i)}$ may be used to characterize all linear codes of $\cI$. Of course $LV^{(S)}Z$ is non-zero if and only if it has positive weight. This result will be generalized further in Theorem \ref{th:err} to give a criterion for error-correction.}

\begin{corollary}\label{cordec}
   Let $L\in \fq^{N \times d_S}$. Then $L$ represents an $\fq$-linear $\cI$-IC of length $N$
   if and only if $\INSe{LV^{(S)}Z \neq 0}$ for each $i\in [m]$, and $Z \in \cZ^{(i)}$.	
\end{corollary}

\begin{proof}
	\INSe{Fix some $i \in [m]$ and} let $Z_0 \in \cZ^{(i)}$, let $LV^{(S)}Z_0 = W$. Suppose that $R_i \notin \cL + \cX^{(i)}. $ 
	%$R_i$ is not contained in the row space of the matrix $\left[ \begin{array}{r} V^{(i)} \\ LV^{(S)} \end{array} \right]$.
	Then as in the proof of Lemma \ref{lemdecode}, the linear system
	\begin{equation}\label{eqcons}
	    R_iZ = U , V^{(i)}Z = 0, LV^{(S)}Z = W 
	\end{equation}
	is consistent for every choice of $U \in \fq^{n}$. In particular, (\ref{eqcons}) has a solution $Z_1$ for $U=0$. %, in which case $Z_1 \in \cW^{(i)} $.
	Then $Z = Z_0-Z_1 \in \cZ^{(i)}$ and $ LV^{(S)}Z = 0$. \INSe{We have shown that if $L$ does not represent a linear code for $\cI$ then for some $i$, there exists $Z \in \cZ^{(i)}$ such that $ LV^{(S)}Z = 0$. Applying the contrapositive, this yields that} if $\rk(LV^{(S)}Z) \geq 1$ \INSe{(i.e. if $LV^{(S)}Z) \neq 0$)}
	for each $i\in [m]$ and $Z \in \cZ^{(i)}$, then $L$ represents a linear index code for the instance $\cI$.
	
	Conversely, if there exist $A \in \fq^{d_i},B \in \fq^N$ such that $R_i = AV^{(i)} + B LV^{(S)}$
	then $$R_iZ = AV^{(i)} Z+ B LV^{(S)} Z = B LV^{(S)}Z \neq 0 ,$$
	%so that $\rk(LV^{(S)}Z) \geq 1 $ 
	for any $Z \in \cZ^{(i)}$.
\end{proof}

\INSe{\subsection{The Optimal Length of an $\fq$-Linear $\cI$-IC}}

We extend the definition of the min-rank of an instance of the ICSI problem, as given in \cite{Son}, to the ICCSI problem. \INSe{We will show that this characterizes the shortest possible length of an $\fq$-linear $\cI$-IC.}

\begin{definition}
 We define the min-rank of the instance $\cI$ of the ICCSI problem over $\fq$ to be
\begin{eqnarray*}
%\gk(\cI) & = &\min\{\rk_q(\{R_i+{\bf v}^{(i)}\}_{i\in[m]})\,:\,{\bf v}^{(i)}\in\cX^{(i)}\}
\gk(\cI) & = &\min\{\rk(A + R) :A \in \fq^{m \times n}, A_i \in \cX^{(i)} \cap \cX^{(S)} \subset \fq^n,\; \forall i \in [m] \}.           
\end{eqnarray*}
\end{definition}
\noindent Observe that the quantity $\gk(\cI)$ is $d_{\text{rk}}(R,\cX \cap \tilde{\cX} )$, which is the rank-distance of 
$R \in \fq^{m \times n}$ to the $\fq$-linear code $\cX \cap \tilde{\cX}$, or equivalently the minimum rank-weight of the coset 
$R + (\cX \cap \tilde{\cX}) \subset \fq^{m\times n}$.

We now show that given the instance $\cI$, the minimum length of an $\fq$-linear $\cI$-IC is given by its min-rank. 

\begin{lemma}\label{lemlength}
The length of an optimal $\fq$-linear $\cI$-IC is $\gk(\cI)$.
\end{lemma}

\begin{proof}
	\INSe{Let $L \in \fq^{N \times d_S}$ have rank $N$. From Lemma \ref{lemdecode}, $L$ represents a linear code of length $N$ if and only if 
		for each $i \in [m]$ there exist $A_i \in \cX^{(i)} \cap \cX^{(S)} \subset \fq^n, B_i \in \fq^N$ such that $$R_i = B_i L V^{(S)}-A_i,$$ (i.e. if and only if $R_i \in\cL+ \cX^{(i)} $ for each $i$). Equivalently this holds if and only if there exist matrices $A \in \cX \cap \tilde{\cX}$, $B \in \fq^{m \times N}$ such that $R = BLV^{(S)}-A$, in which case we have $BLV^{(S)}=R+A$ in the coset $R+(\cX \cap \tilde{\cX})$. In particular, we have shown that every matrix $L \in \fq^{N \times d_S}$ represents
		an $\fq$-linear code for $\cI$ only if $$BLV^{(S)} \in R+(\cX \cap \tilde{\cX})$$ for some $B \in \fq^{m \times N}$, so every such $L$ has rank at least $\gk(\cI)$. }
	%Let $\cI=(t,m,n,\cX,\cX^{S},R)$ be an instance of the ICCSI problem $\fq$. 
    %Let $N$ be a positive integer and let $L \in \fq^{N\times d_S}$ represent a  linear code for $\cI$. 
    
    Now let $A \in \fq^{m\times n}$ with $A_i \in \cX^{(i)} \cap \cX^{(S)}$ for each $i \in [m]$.
    % such that $\rk(A+R) = \gk(\cI)$.
    Suppose that $A+R$ has rank $N$.  
    Since $A,R \in  \tilde{\cX}$, there exists $Z\in \fq^{m \times d_S}$ of rank $N$ satisfying $A+R = Z V^{(S)}$.
    Furthermore, there exist $B \in \fq^{m\times N}$ and $L \in \fq^{N\times d_S}$ such that $Z=BL$.
    Then $$R=A-BLV^{(S)}$$ so $L$ represents a  linear code of length $N$ 
    for the instance $\cI$.
    The length $N$ is \INSe{minimized} for $N = \gk(\cI)$, \INSe{so there exists some $L$ of rank $\gk(\cI)$ representing a linear code for $\cI$.
    	} 
\end{proof}

\INSe{Lemma \ref{lemlength} gives a naive algorithm for computation of a matrix $L$ for an optimal linear $\cI$-IC: put each element of $R+(\cX \cap \tilde{\cX})$ into row-echelon form and choose one of minimal rank $N=\gk(\cI)$. The non-zero rows of this matrix yields the required $N \times d_S$ matrix $L$. We do not suggest this as a practical approach, since it requires ${\cal O}(d_S^3 q^\ell)$ operations, with $\ell = \dim \cX \cap \tilde{\cX}$. We mention this here to give a concrete illustration of the realization problem.}
As already observed in \cite{Son1}, the min-rank $\gk(\cI)$ of the instance $\cI$ generalizes the notion of the min-rank of the so-called
side-information graph of the classical index coding problem, which is NP-hard to compute. A discussion on the various approaches to obtaining bounds on the optimal length of an index code can be read in \cite{SDL}, where the authors assert that graph-theoretic methods for constructing index coding schemes yield bounds on the optimal length of an index code, which are often out-performed by the min-rank. 
\INSe{In fact all of these so called graph-theoretic methods, which use linear programming methods to obtain (possibly sub-optimal) solutions to the linear index coding problem can be extended to the ICCSI case. These results have been outlined in a separate forthcoming paper \cite{BC}.}
%In \cite{Son1} an inexplicit upper bound is given on the optimal length of the ECIC's; the proof is based on the construction of a random ECIC. In this section we extend the construction to the ICCSI problem case. Moreover we obtain a bound on the optimal length of an IC that is tighter than the bound derived from Theorem $6.1$ in \cite{Son1}.

\INSe{\subsection{Upper Bounds on the Optimal Length of an $\fq$-Linear $\cI$-IC}}

\INSe{We now give upper bounds on $\gk(\cI)$, applying probabilistic arguments. The main results are Corollary \ref{corzip} and Theorem \ref{thbd2}, both of which show that with certain constraints on $N$, there exists an $\fq$-linear $\cI$-IC of length $N$. While both results essentially give lower bounds on the probability that a random $N \times d_S$ matrix $L$ represents an $\fq$-linear $\cI$-IC, the key point is that these probabilities are positive, so that existence is guaranteed. It is from this observation that upper bounds on $\gk(\cI)$ are achieved.}

We will use the following theorem proved by Zippel \cite{zippel} (see also \cite{demillo,schwartz}).
We state it here for finite fields.

\begin{theorem}\label{lm:ko}
	Let $m,s$ be positive integers with $q>m$ and let $P(x_1,...,x_s)$ be a non-zero multivariate polynomial in $\fq[x_1,...,x_s]$ for which the largest exponent of any variable $x_i$ is at most $m$. If $(a_1,...a_s)$ is chosen uniformly at random in $\fq^s$ then the probability that $P(a_1,...,a_s)$ equals zero is at most $1-(1-m/q)^s$.
\end{theorem}

\begin{remark}\label{iir}
	Before proving the following theorem, we note that if $X_1,\dots, X_n$ are independent uniformly distributed random variables that take their values over a field $\fq$, then the random variable
	$$
	%Z=\sum_{i=1}^n \ga_i X_i,
	Z_\ell=\sum_{i=1}^\ell \ga_i X_i,
	$$
	for some $\ell \in [n]$, $\ga_i\in \fq^\times \INSe{= \fq \backslash \{0\}}$, has a uniform distribution. 
	
	This is easily shown by an inductive argument. Clearly $P(Z_1=\beta) = \frac{1}{q}$ for any $\beta \in \fq$ since $\alpha_1 \neq 0$.
	Moreover, for any $\ell \in [n],\beta \in \fq$, 
	\begin{eqnarray*}
		P(Z_{\ell} = \beta) & = & P(Z_{\ell-1} = \beta - \alpha_\ell X_\ell )\\
		& = & \sum_{\gamma \in \fq}P(X_{\ell} =\gamma)P(Z_{\ell-1} = \beta - \alpha_\ell \gamma) = \frac{1}{q}.
	\end{eqnarray*}
	
\end{remark}

\INSe{Let $m'$ be the number of distinct equivalence classes of $[m]$ under the relation $i \equiv j$ if $\cX^{(i)}=\cX^{(j)}$.
	Let $\tilde{m}$ be a set of $m'$ representatives for the distinct equivalence classes of $[m]$.}

\begin{theorem}\label{th:random1}
	Let $\cI$ be an instance of an ICCSI problem and let $N = \max\{n-d_i :i\in[m]\}$. Suppose that $\INSe{q>m'}$.
	If the entries of a matrix $L\in\fq^{N\times d_S}$ are chosen uniformly at random in $\fq$, then the probability that $L$ represents a linear code for $\cI$ is at least $(1-m'/q)^{Nd_S}$. 
\end{theorem}

\begin{proof}
	From Corollary \ref{cordec}, if $w\left(LV^{(S)}Z\right)\ge 1$ for each $Z \in \cY^{(i)}$ then $L$ represents a code for $\cI$. 
	%Note that if $w\left(LV^{(S)}Z\right)\ge 1$ for each $Z$ in $\cY^{(i)}$ then in particular $V^{(S)}Z \neq 0$ for each such $Z$ so that
	%${V^{(i)}}^{\perp} \cap {V^{(S)}}^{\perp}= \{0\}$.
	%This then implies that $n-d_i \leq d_S$ for each $i \in [m]$ 
	\INSe{For each $i \in \tilde{m}$,} let $Z^{(i)} \in \fq^{n \times k_i}$ satisfy $V^{(i)}Z^{(i)} =0$
	and have rank $k_i=n-d_i$. Write $L^{(i)}=LV^{(S)}Z^{(i)}$. The matrix $L$ represents a code for $\cI$ if $L^{(i)}$ is a full-rank matrix for each $i \in \INSe{\tilde{m}}$, which holds if and only if there exists a non-zero $k_i\times k_i$ minor $M^{(i)}$ of $L^{(i)}$. Since the entries of $L$ are uniformly distributed, so are the entries of $L^{(i)}$, from Remark \ref{iir}. 
	%Now, from Remark \ref{iir} and from the fact that there are no zero columns in $M_i$, we have that the entries of $LM_i$ are distributed uniformly at random over the field $\fq$.
	Each such minor has the form
	$M^{(i)}=\sum_{\sigma \in S_{k_i}} sgn(\sigma) \prod_{j=1}^{k_i} L^{(i)\sigma(v_j)}_{v_j} .$
	Now $\prod_{i \in \INSe{\tilde{m}}} M^{(i)}$ may be viewed as a polynomial in $Nd_S$ variables of degree 
	$\displaystyle{\sum_{i \in \INSe{\tilde{m}}} k_i \leq \INSe{m'}N}$ with each variable appearing with multiplicity
	at most $\INSe{m'}$ in any term. 
	Then the probability that $L$ represents a code for $\cI$ is the probability that 
	$\prod_{i \in \INSe{\tilde{m}}} M^{(i)}$ is non-zero, which from Lemma \ref{lm:ko}
	is at least $(1-\INSe{m'}/q)^{Nd_S}$, for $q>\INSe{m'}$.
\end{proof}

\begin{corollary}\label{corzip}
   If $q>\INSe{m'}$ then $\gk(\cI) \leq \max \{n-d_i :i\in[m]\}$.
\end{corollary}	

\begin{proof}
	Theorem \ref{th:random1} guarantees the existence of some matrix $L \in \fq^{N\times d_S}$ that represents an $\fq$-linear
	$\cI$-IC of length $N=\max \{n-d_i :i\in[m]\}$. The result is now immediate since $N \geq \gk(\cI)$.
\end{proof}

\begin{remark}
	In fact Schwartz's result \INSe{\cite{schwartz}} gives the lower bound of $\INSe{1-\frac{m'(n-d)}{q}}$ on the probability of an $N \times d_S$ matrix $L$
	representing an $\cI$-IC, where $d$ is the average of the $\{d_i:i \in [m]\}$. While this may give a higher lower bound,
	it places the restriction $\INSe{m'(n-d) <q}$ and so in particular yields a weaker version of Corollary \ref{corzip}.  
\end{remark}	

%\begin{remark}
%	\INSe{Theorem \ref{th:random1} and its corollary may be viewed as the index coding analogue of \em{The Main Network Coding Theorem} \cite[Theorem 2.2]{FS}.} 
%\end{remark}

\begin{remark}
	\INSe{Note that if for some $i$, $LV^{(S)}Z$ is non-zero for any $Z \in \cY^{(i)}$, then it satisfies the decoding criterion for any possible request vector $R_i \in \cX^{(S)}$, and hence delivers all possible requests to User $i$. Therefore, Theorem \ref{th:random1} and Corollary \ref{corzip} should be viewed in the context of similar results in \cite{BK, tehrani2012bipartite}, which lead to {\em partial clique-cover} and {\em partition multicast} schemes. There is also a close association with the so-called \em{Main Network Coding Theorem} \cite[Theorem 2.2]{FS} for multicast network coding.
	All of these results rely on the field size $q$ being sufficiently large to invoke Zippel's theorem and its variants. 
	} 
\end{remark}

\begin{remark}
\INSm{
The approach in \cite{BK} to construct a linear IC for a partial clique is based on \INSe{{\em maximum distance separable}} (MDS) codes (this can be used also in the more general case of a multicast group, as described in \cite{tehrani2012bipartite}). Any generator matrix of an MDS linear code of length $n$ and dimension $k$ is such that any $k$ columns are linear independent. \INSe{Suppose that $\cI$ is an ICSI instance and that $d_i$ is the number of uncoded packets $X_j$ known to the receiver $i$. Let} $G$ be a generator matrix of an MDS code of length $n$ and dimension $N=\max \{n-d_i :i\in[m]\}$. Then the sender can broadcast the following \INSe{linear} combination of the columns of $G$:
$$
X_1G^1+...+X_nG^n.
$$
Without loss of generality, suppose that some receiver $i$ has $X_{N+1},...,X_n$, \INSe{and can thus} recover 
$$
X_1G^1+...+X_NG^N.
$$
From the MDS property of $G$, \INSe{the first $N$ columns of $G$ form an invertible matrix so that the user can determine} $(X_1,...,X_N)$.
\INSe{In terms of Lemma \ref{lemdecode}, the side-information is encoded by a matrix $V^{(i)}=[0|I]$ whose rows are standard basis vectors. Appending
	the $N$ rows of $G$ then results in a matrix with row space $\fq^n$.
In the above we get a matrix
	$$ \left[\begin{array}{cc} 
	         G^{[N]} & G^{[n]\backslash [N]}\\
	         0       &  I
	         \end{array}\right], $$
 which has rank $n$, so any possible request vector is contained in $\langle G\rangle+\cX^{(i)}$.}\\
However, this approach with MDS codes is \INSe{not} possible \INSe{in} the more general case of coded side information.\\
Indeed, suppose we have $n=m=4$, $\cX^{(S)}=\FF_2^4$, the user side-information determined by	
	$$
	V^{(1)} = \left[\begin{array}{cccc} 
	1 & 0 & 0 & 1\end{array}\right],
	V^{(2)} = \left[\begin{array}{cccc} 
	0 & 0 & 0 & 1
	\end{array}\right],
	V^{(3)} = \left[\begin{array}{cccc} 
	0 & 1 & 0 & 0
	\end{array}\right],    
	V^{(4)} = \left[\begin{array}{cccc} 
	0 & 0 & 1 & 0
	\end{array}\right],
	$$
and requests 
	$$R_1=[1110],R_2=[0100],R_3=[0010],R_4=[0001].$$
If the approach using an MDS code could be applied, then the matrix 
$$
G=\left[\begin{array}{cccc}
1&0&0&1\\
0&1&0&1\\
0&0&1&1\end{array}\right]
$$
defined over $\mathbb{F}_2$, could be used to encode a vector $X=[X_1,...,X_4]^T$. It is easy to check that $R_1\notin\langle G\rangle+\cX^{(1)}$ and so from Lemma \ref{lemdecode} we have that Receiver $1$ cannot decode. On other hand consider the matrix
$$
L=\left[\begin{array}{cccc}
1&0&0&0\\
0&1&0&1\\
0&0&1&1\end{array}\right],
$$
defined over $\FF_2$. Following the proof of Theorem \ref{th:random1} we have the matrices 
$$
Z^{(1)}=\left[\begin{array}{ccc}
1&0&0\\
0&1&0\\
0&0&1\\
1&0&0\end{array}\right],
Z^{(2)}=\left[\begin{array}{ccc}
1&0&0\\
0&1&0\\
0&0&1\\
0&0&0\end{array}\right],
Z^{(3)}=\left[\begin{array}{ccc}
1&0&0\\
0&0&0\\
0&1&0\\
0&0&1\end{array}\right],
Z^{(4)}=\left[\begin{array}{ccc}
1&0&0\\
0&1&0\\
0&0&0\\
0&0&1\end{array}\right].
$$
It is easy to check that for each $i$ $LZ^{(i)}$ is a $3\times 3$ matrix of full rank. Then the matrix $L$ represents a linear index code for our instance but $L$ does not generate an MDS code.}
\INSe{In fact, direct inspection shows that $\cL +\cX^{(i)}=\FF_2^4$ for each $i$ so that $L$ realizes an $\FF_2$-linear $\cI$-IC for any choice of $R$. This is equivalent to multicast network coding.
}
\end{remark}

We can remove the explicit constraint that $q>m'$ to obtain an alternative lower bound on the probability of a realizable \INSe{linear} solution to the ICCSI problem.
A straightforward counting argument yields the following \INSe{`folklore' result, a proof of which we include for completeness. 
Recall that for positive integers $s \geq r$ the Gaussian coefficient 
$$\qbin{s}{r}{q} := \frac{\prod_{j=0}^{r-1} (q^{s}-q^j)}{\prod_{j=0}^{r-1} (q^{r}-q^j) }$$ 
denotes the number of $r$-dimensional subspaces contained in an $s$-dimensional space over $\fq$. If \INSm{$s < r$} this number is zero.} 

\begin{lemma}\label{lemssp}
	Let $W,V,S$ be subspaces of $\fq^n$ with $W<V\cap S$ and of dimensions $w,v$ and $s$ respectively.
	Suppose that $S \cap V$ has dimension $\ell$.
	The number of $N$-dimensional subspaces $U$ of $S$ satisfying $V\cap U \subset W$ is
	$$\sum_{r=0}^{w} q^{(\INSe{\ell}-r)(N-r)} \qbin{w}{r}{q} \qbin{s-\ell}{N-r}{q}.$$
\end{lemma}

\begin{proof}
	Let $M$ be an $r$-dimensional subspace of $V \cap S$.
	A basis $\INSe{\{m_1,...,m_r\}}$ of $\INSe{M}$ can be completed to a linearly independent $N$-set 
	by appending some $\INSe{m_{r+1},...,m_N} \subset S \backslash V$ in
	$$\prod_{j=\ell}^{\ell+N-r-1}(q^{s}-q^j) =q^{\ell(N-r)}\prod_{j=0}^{N-r-1}(q^{s-\ell}-q^{j}) $$
	ways. There are $(q^N-q^r)\cdots (q^N-q^{N-1})$ choices of $\INSe{m_{r+1},...,m_N}$ in 
	$M + \langle \INSe{m_{r+1},...,m_N} \rangle \backslash M$. Therefore there are \\
	$q^{(\ell-r)(N-r)}\qbin{s-\ell}{N-r}{q}$ $N$-dimensional subspaces of $S$ that meet $V$ in a given
	$r$-dimensional subspace of $V$.
	The result now follows since there are $\qbin{w}{r}{q}$ $r$-dimensional subspaces of $W$. 
\end{proof}

\INSe{We'll now apply Lemma \ref{lemssp} for the case $W=\cX^{(i)}, V= \langle R_i, \cX^{(i)} \rangle$ to count the number of subspaces $\cL \in \cX^{(S)}$ contained in their intersection. This will tell us the number of subspaces $\cL$ such that $R_i \notin \cL + \cX^{(i)}$.}

\begin{theorem}\label{thbd2}
	Let $\cI$ be an instance of an ICCSI problem. For each $i \in [m]$, let $\dim (\cX^{(i)}\cap \cX^{(S)}) = w_i$.
	The probability that there exists an $N$-dimensional subspace $\cL$ of $\cX^{(S)}$ such that for each $i \in [m]$,
	$R_i \in \cL + \cX^{(i)}$, is at least
	$$1 - \qbin{d_S}{N}{q}^{-1} \sum_{i=1}^m \sum_{r=0}^{w_i} q^{(w_i+1-r)(N-r)} \qbin{w_i}{r}{q} \qbin{d_S-w_i-1}{N-r}{q}. $$
	In particular, there exists a linear $\cI$-IC of length $N$ if 
	$$ \sum_{i=1}^m \sum_{r=0}^{w_i} q^{(w_i+1-r)(N-r)} \qbin{w_i}{r}{q} \qbin{d_S-w_i-1}{N-r}{q} <  \qbin{d_S}{N}{q} .$$	
\end{theorem}

\begin{proof}
		Let $i \in [m]$. An $N$-dimensional subspace $\cL$ of $\cX^{(S)}$ satisfies 
		$R_i \in \cL + \cX^{(i)}$ if and only if $(R_i + \cX^{(i)}) \cap \cL$ \INSe{is} non-empty.
		The number of $N$-dimensional subspaces of $\cX^{(S)}$, that miss $R_i + \cX^{(i)}$ is the number of $\cL$ 
		that meet $\langle R_i, \cX^{(i)} \rangle$
		in a subspace of $\cX^{(i)}$. Since $R_i \in \cX^{(S)} \backslash \cX^{(i)}$ by assumption, $\langle R_i, \cX^{(i)} \rangle \cap \cX^{(S)}$ has dimension $w_i+1$. 
		\INSe{Then from Lemma \ref{lemssp}, the number of $N$-dimensional subspaces $\cL$ in $\cX^{(S)}$ that miss $R_i + \cX^{(i)}$ (i.e. the number satisfying $R_i \notin \cL + \cX^{(i)}$), is 
			$$ \sum_{r=0}^{w_i} q^{(w_i+1-r)(N-r)} \qbin{w_i}{r}{q} \qbin{d_S-w_i-1}{N-r}{q}.$$
			The probability that an arbitrary $N$-dimensional subspace of $\cX^{(S)}$ misses $R_i + \cX^{(i)}$ is this number divided by 
			$ \qbin{d_S}{N}{q}$.
		    From the union bound, the probability that an arbitrary $N$-dimensional subspace of $\cX^{(S)}$ misses every $R_i + \cX^{(i)}$ is upper bounded
		    by 	$$\qbin{d_S}{N}{q}^{-1} \sum_{i=1}^m \sum_{r=0}^{w_i} q^{(w_i+1-r)(N-r)} \qbin{w_i}{r}{q} \qbin{d_S-w_i-1}{N-r}{q}. $$}	 
\end{proof}

\begin{remark}
	\INSe{In the above argument, by invoking the union bound, we have assumed the most extreme case. 
		For each $i$, let $S^{(i)}$ denote the set of $N$-dimensional subspaces $\cL$ of $\cX^{(S)}$ such that 
		$$\cX^{(i)} \cap \cL= \langle R_i, \cX^{(i)} \rangle \cap \cL.$$ 
		Then the probability of a decoding failure is maximized 
		when the size of the union of the $S^{(i)}$ is maximized.
		If the $S^{(i)}$ are pairwise disjoint we have 
			$$\left|\;\cup_{i \in [m]} S^{(i)}\;\right| = \sum_{i \in [m]} |S^{(i)}|,$$
	     and so the given bound is sharp. This occurs if no $N$-dimensional subspace is contained in the intersection of any pair of the $S^{(i)}$. Moreover, there is no length $N$ $\fq$-linear $\cI$-IC if and only if 
	     $\cup_{i \in [m]} S^{(i)}$ contains all $N$-dimensional subspaces of $\cX^{(S)}$.}  
\end{remark}

\begin{remark}
	Given an $N$-dimensional subspace $\cL < \cX^{(S)}$ 
	Theorem \ref{thbd2} often yields a better lower bound on the probability that $L \in \fq^{N \times d_S}$ 
	satisfying $\cL=\langle LV^{(S)} \rangle $ represents a linear $\cI$-IC.
	Furthermore, it establishes existence for the case $m \geq q.$ 	
\end{remark}

\begin{example}
	Let $m=6,n=4,q=2$ and let $\cX^{(S)}=\fq^4$.
	Suppose that $\cX^{(i)}$ has dimension $d_i=2$ for each $i \in \{1,...,6\}$.
	According to Theorem \ref{thbd2}, the probability of the existence of a $3 \times 4$ matrix 
	$L$ over $\FF_2$ that represents a linear $\cI$-IC with these parameters is at least 0.2,
	but is inconclusive for the existence of a $2 \times 4$ encoding matrix $L$.
	Indeed there are such $\cI$-IC satisfying $\kappa(\cI)=3\;\INSe{>2=\max\{n-d_i:i \in [6]\}}$, for example, $\cI$ with 
	user side-information determined by	
	$$
	V^{(1)} = \left[\begin{array}{cccc} 
	0 & 0 & 1 & 0 \\
	0 & 0 & 0 & 1
	\end{array}\right],
	V^{(2)} = \left[\begin{array}{cccc} 
	1 & 0 & 0 & 0 \\
	0 & 0 & 0 & 1
	\end{array}\right],
	V^{(3)} = \left[\begin{array}{cccc} 
	1 & 0 & 0 & 0 \\
	0 & 1 & 0 & 0
	\end{array}\right],
	$$
	$$         
	V^{(4)} = \left[\begin{array}{cccc} 
	0 & 1 & 0 & 0 \\
	0 & 0 & 1 & 0
	\end{array}\right],
	V^{(5)} = \left[\begin{array}{cccc} 
	1 & 0 & 0 & 0 \\
	0 & 0 & 1 & 0
	\end{array}\right],
	V^{(6)} = \left[\begin{array}{cccc} 
	0 & 1 & 0 & 0 \\
	0 & 0 & 0 & 1
	\end{array}\right], 
	$$
	and requests 
	$$R_1=[1000],R_2=[0100],R_3=[0010],R_4=[0001],R_5=[0100],R_6=[1000]$$
	has min-rank equal to $3$.
	
	On the other hand, there are several examples of $\cI$ for the same parameters that have min-rank equal to 2,
	such as that defined by:
	$$
	V^{(1)} = \left[\begin{array}{cccc} 
	1 & 0 & 1 & 0 \\
	0 & 0 & 0 & 1
	\end{array}\right],
	V^{(2)} = \left[\begin{array}{cccc} 
	1 & 0 & 0 & 0 \\
	0 & 0 & 1 & 1
	\end{array}\right],
	V^{(3)} = \left[\begin{array}{cccc} 
	1 & 0 & 0 & 0 \\
	0 & 1 & 0 & 0
	\end{array}\right],
	$$
	$$         
	V^{(4)} = \left[\begin{array}{cccc} 
	0 & 1 & 0 & 1 \\
	0 & 0 & 1 & 0
	\end{array}\right],
	V^{(5)} = \left[\begin{array}{cccc} 
	1 & 0 & 1 & 0 \\
	0 & 0 & 0 & 1
	\end{array}\right],
	V^{(6)} = \left[\begin{array}{cccc} 
	0 & 1 & 0 & 0 \\
	0 & 0 & 0 & 1
	\end{array}\right], 
	$$
	and requests 
	$$R_1=[1100],R_2=[0111],R_3=[1010],R_4=[1001],R_5=[0100],R_6=[1111].$$
\end{example}	

We tabulate (see Table II) evaluations of the lower bound on the probability of the existence of linear $\cI$-ICs with $\cX^{(S)}=\FF_4^{10}$
and length $N=10-d$ where $d_i=d$ for each $i$th user in $[m],m\;\INSe{=m'}>1$.
In each case, we find that the maximal value of $m'=m$ for which the bound of Theorem \ref{thbd2} can be applied is $m=q$. However,
existence for $m>q$ can be established by Theorem \ref{thbd2} for $N>n-d$. In Table III we record parameters $N$ and $d=10-N+1$ of $\cI$-IC known to exist using the bound of Theorem \ref{thbd2} for a maximal number of users $m$.

\begin{table}
\begin{center}
\begin{tabular}{|c|c|c|c|c|c|}
	\hline
	  $q$ & $N$ & $d$  & $m$ & Theorem \ref{th:random1} & Theorem \ref{thbd2}  \\
	\hline  
	  4 & 1 & 9 & 2 & \INSm{0.001}                   & \INSm{$0.500$} \\
	  4 & 1 & 9 & 3 & \INSm{$9.536 \times 10^{-7}$}  & \INSm{$0.250$} \\ 
	  4 & 1 & 9 & 4 & $- $                           & \INSm{$2.861 \times 10^{-6}$} \\
	  4 & 2 & 8 & 2 & \INSm{$9.536\times 10^{-7} $}  & \INSm{$0.523 $}\\
	  4 & 2 & 8 & 3 & \INSm{$9.094\times 10^{-13}$}  & \INSm{$0.285 $}\\
	  4 & 2 & 8 & 4 &  $- $                          & \INSm{$0.046 $}\\
	  4 & 3 & 7 & 2 & \INSm{$9.313\times 10^{-10}$}  & \INSm{$0.528 $}\\
	  4 & 3 & 7 & 3 & \INSm{$8.673\times 10^{-19}$}  & \INSm{$0.293 $}\\
	  4 & 3 & 7 & 4 & $- $                           & \INSm{$0.057 $}\\
	  4 & 3 & 6 & 2 &\INSm{$9.094\times 10^{-13}$}  & \INSm{$0.530 $}\\
	  4 & 4 & 6 & 3 &\INSm{$8.271\times 10^{-25}$}   & \INSm{$0.295 $}\\
	  4 & 4 & 6 & 4 &$- $                            & \INSm{$0.060 $}\\
	  4 & 5 & 5 & 2 &\INSm{$8.881\times 10^{-16} $}  & \INSm{$0.532 $}\\
	  4 & 5 & 5 & 3 &\INSm{$7.888\times 10^{-31} $}  & \INSm{$0.296 $}\\
	  4 & 5 & 5 & 4 &$- $                            & \INSm{$0.061$}\\
	  4 & 6 & 4 & 2 &\INSm{$8.673\times 10^{-19} $} & \INSm{$0.532$}\\
	  4 & 6 & 4 & 3 &\INSm{$7.523\times 10^{-37}$}   & \INSm{$0.298$}\\
	  4 & 6 & 4 & 4 &$- $                            & \INSm{$0.064$}\\
	  4 & 7 & 3 & 2 &\INSm{$8.470\times 10^{-22} $} & \INSm{$0.536$}\\
	  4 & 7 & 3 & 3 &\INSm{$7.174\times 10^{-43} $} & \INSm{$0.304$}\\
	  4 & 7 & 3 & 4 &$- $                            & \INSm{$0.072$}\\
	  4 & 8 & 2 & 2 &\INSm{$8.271\times 10^{-25} $} & \INSm{$0.553$}\\
	  4 & 8 & 2 & 3 &\INSm{$6.842\times 10^{-49} $} & \INSm{$0.329$}\\
	  4 & 8 & 2 & 4 &$- $                            & \INSm{$0.106$}\\
	  4 & 9 & 1 & 2 &\INSm{$8.077\times 10^{-28} $} & \INSm{$0.624$}\\
	  4 & 9 & 1 & 3 &\INSm{$6.525\times 10^{-55} $}  & \INSm{$0.437$}\\
	  4 & 9 & 1 & 4 &$- $                            & \INSm{$0.249$}\\
	  4 & 9 & 1 & 5 &$- $                            & \INSm{$0.062$}\\
	 \hline
	 
\end{tabular}
\caption{}
\end{center}

\end{table}

\begin{table}
	
	\begin{center}
		\begin{tabular}{|c|c|c|c|c|}
			\hline
			$q$ & $N$ & $d$  & $m$ & Theorem \ref{thbd2}  \\
			\hline
			4 & 2 & 9 & 16 & \INSm{$1.4305 \times 10^{-5} $} \\
			4 & 3 & 8 & 16 & \INSm{$0.0117$}\\
			4 & 4 & 7 & 16 & \INSm{$0.0148$}\\
			4 & 5 & 6 & 16 & \INSm{$0.0162$}\\
			4 & 6 & 5 & 16 & \INSm{$0.0191$}\\
			4 & 7 & 4 & 16 & \INSm{$0.0299$}\\
			4 & 8 & 3 & 17 & \INSm{$0.0155$}\\
			4 & 9 & 2 & 21 & \INSm{$0.0156$}\\
			8 & 2 & 9 & 64 & \INSm{$5.8673\times 10^{-8}$} \\
			8 & 3 & 8 & 64 & \INSm{$0.0017$}\\
			8 & 4 & 7 & 64 & \INSm{$0.0019$}\\
			8 & 5 & 6 & 64 & \INSm{$0.0019$}\\
			8 & 6 & 5 & 64 & \INSm{$0.0021$}\\
			8 & 7 & 4 & 64 & \INSm{$0.0038$}\\
			8 & 8 & 3 & 65 & \INSm{$0.0019$}\\
			8 & 9 & 2 & 73 & \INSm{$0.0019$}\\
			16& 2 & 9 & 256	&\INSm{$2.3192\times 10^{-10}$} \\
			16& 3 & 8 & 256 &\INSm{$0.0002$}\\
			16& 4 & 7 & 256 &\INSm{$0.0002$}\\
			16& 5 & 6 & 256 & \INSm{$0.0002$}\\
			16& 6 & 5 & 256 & \INSm{$0.0002$}\\
			16& 7 & 4 & 256	& \INSm{$0.0004$}\\
			16& 8 & 3 & 257 & \INSm{$0.0002$}\\
			16& 9 & 2 & 273 &\INSm{$0.0002$}\\
			\hline       
		\end{tabular}
		\caption{}
	\end{center}
\end{table}

\section{Error Correction in the ICCSI Problem}

We now discuss error-correction in the ICCSI problem, extending the ideas presented in \cite{Son1}, \INSe{in} two ways. The first direction is in the context of coded-side information, as presented in \cite{shum,dai}. The second allows for error correction for the rank metric, in the transmission of matrices in $\fq^{N \times t}$ when $t>1$.  
For the remainder, we let $\cM \subset \fq^{n\times t}$ denote the message space associated with the ICCSI problem. 

\INSe{\subsection{The ECIC Problem}}

\begin{definition}
	Let $\cI$ be an instance of an ICCSI problem and let $N$ be a positive integer. We say that the map	
	$$E:\fq^{n\times t}\to\fq^{N \times t},$$
	is a $\gd$-error correcting code for $\cI$ of length $N$, and say that $E$ is an $(\cI,\gd)$-ECIC, if for each $i$th receiver there exists a decoding map	
	$$	D_i:\fq^{N\times t}\times\cX^{(i)}\to\fq^t,$$
	satisfying
	$$ D_i(E(X)+W,A)=R_iX $$
	for all $X\in\cM$ and $W \in \fq^{N\times t},\, w(W)\le\gd$ for some vector $A\in \cX^{(i)}$. 
	$E$ is called a linear code for $\cI$, or an $\fq$-linear $\cI$-ECIC if $E(X)=LV^{(S)}X$ for some $L\in\fq^{N \times d_S}$, in which case we say that $L$ represents the linear $(\cI,\gd)$-ECIC $E$.
\end{definition}

The standard coding theory argument gives a criterion for the existence of a $(\cI,\gd)$-ECIC, extending \cite[Lemma 3.8]{Son1}. \INSe{The following result is the error correction analogue of Lemma \ref{lemdecode}. It basically says that the $i$th receiver can correct up to $\delta$ errors and uniquely decode its requested packet $R_iX$ if any pair of confusable data matrices have encodings $LV^{(S)}X$ and $LV^{(S)}X'$ that are at distance at least $2\gd+1$ apart.} 

\begin{theorem}\label{th:err}
	Let $\cI$ be an instance of an ICCSI problem and let $N$ be a positive integer.
	A matrix $L\in\fq^{N\times d_S}$ represents a linear $(\cI,\gd)$-ECIC if and only if for all $i\in[m]$ it holds that
	$$
	w\left(LV^{(S)}(X - X')\right)\ge 2\gd+1,
	$$
	for all $X,X' \in \cM$ such that $X - X' \in \cZ^{(i)}$.
	%\begin{enumerate}
	%	\item $Z \in \cZ^{(i)}$ if $t=1$ and $w$ denotes the Hamming metric,
	%	\item $Z\in \cZ_\gd^{(i)}$ if $t>1$ and $w$ denotes the rank metric.
	%\end{enumerate} 
\end{theorem}
\begin{proof}
	Let $L\in\fq^{N\times d_S}$ represent a linear $\cI$-IC.
	For each $X \in \fq^{n \times t}$, define
	$$B(X,\gd)=\{Y\,:\,Y=LV^{(S)}X+W, \, W\in\fq^{N\times t},\, w(W)\le\gd\}.$$
	\INSe{It is not hard to see by an adaptation of the usual coding theory arguments that} the $i$th receiver can correct $\gd$ errors if and only if
	$$B(X,\gd)\cap B(X',\gd)=\emptyset$$
	for each $X,X' \in \cM$ such that 
	$V^{(i)}X =V^{(i)}X' $ and $R_iX \neq R_iX'$.
	
	%\INSe{For the Hamming metric, the argument is almost identical to that for a classical error-correcting code. 
	%Suppose then that $t>1$ and that $w$ measures the rank weight. Suppose that $L$ represents a linear $(\cI,\gd)$-ECIC. Now let $X,X' \in \cM$ such that $X - X' \in \cZ^{(i)}$.
	%Let $LV^{(i)}X = Z$ and let $LV^{(i)}X' = Z'$. 
	%We may assume that $A=Z-Z'$ is in row-echelon form of rank $d>0$. Now construct a matrix $B$ by modifying $A$ as follows. Replace the zero rows of $A$ with the corresponding rows  of $Z$. Replace the first $\delta$ rows of $A$ with the corresponding rows of $Z'$ and the remaining $d-\delta$ with the corresponding rows of $Z$. Then $d = w(Z-B) + w(Z'-B) = \delta + (d-\delta) > 2 \delta$, by the assumption that the code represented by $L$ is $\delta$ error correcting.}	
	%In particular, the $i$th user can correct $\gd$ errors if and only if for each $i \in [m]$
	%$$w(LV^{(S)}(X-X'))\ge2\gd+1$$
	%whenever $X-X'\in \cZ^{(i)}$.
	
	\INSm{
		For the Hamming metric, the argument is almost identical to that for a classical error-correcting code. 
		Suppose then that $t > 1$ and that $w$ measures the rank weight.\\
		Let $X,X' \in \cM$ such that $X - X' \in \cZ^{(i)}$.
		Let $LV^{(i)}X = Z$ and let $LV^{(i)}X' = Z'$. 
		Let $A=Z-Z'$ and suppose $w(A)=d\le 2\gd$. We may assume that $A$ is in row-echelon form. Thus we can write $A=W+W'$ with $w(W)=\gd$ and $w(W')=d-\gd\le \gd$, where the first $\gd$ rows of $W$ are the corresponding rows of $A$ and the others are zero, and the rows indexed by $[d]\backslash [\gd]$ in $W'$ are the corresponding rows of $A$ and the remaining $N-d+\gd$ rows are zero. That is
		$$A=\left[\begin{array}{c}
		A_{[\gd]}\\
		A_{[d]\backslash [\gd]}\\
		0\end{array}\right], \, 
		W=\left[\begin{array}{c}
		A_{[\gd]}\\
		0
		\end{array}\right]
		\mbox{ and }
		W'=\left[\begin{array}{c}
		0\\
		A_{[d] \backslash [\gd]}\\
		0\end{array}\right].$$
		Then $Z-W=Z'+W'\in B(X,\gd)\cap B(X',\gd)$, so if the requires spheres $ B(X,\gd)$ are disjoint, $L$ represents a linear $\cI$-ECIC.\\
		Conversely, if $B(X,\gd)\cap B(X',\gd)\ne \emptyset$ then $Z+W=Z'+W'$ for some $W$ and $W'$, each having rank at most $\gd$. 
		Thus $Z-Z'=W'-W$ and in particular, by the triangular inequality, $w(Z-Z')\le w(W')+w(W)\le 2\gd$.
	}

\end{proof}

Let $\cI$ be an instance of the ICCSI problem and let $i \in [m]$.
Let $X,X' \in \cM$ such that $X-X' \in \cZ^{(i)}$. 
For the case $t>1$, for any $L \in \fq^{N\times d_S}$ we have
$$w(LV^{(S)}(X-X')\INSe{)} = \rk(LV^{(S)}(X-X')\INSe{)} \leq \rk (V^{(S)}(X-X'))\leq w(X-X').$$
Then $L$ does not represent a linear $(\cI,\gd)$-ECIC if $\rk(X-X') \leq 2\gd$.  
We therefore assume, for $t>1$, that $\cM$ is a subset of $\fq^{n \times t}$ of minimum rank distance at least $2\gd+1$, 
and furthermore that $\{ V^{(S)}X : X \in \cM \} \subset \fq^{d_S \times t}$ has minimum rank distance $2\gd+1$.
Delsarte's result \cite[Theorem 5.4]{D78} yields that $|\cM| \leq q^{(d_S-2\gd)t}$. We define $\cM^\triangle:=\{X-X' : X,\INSm{X'} \in \cM\}$.
For the \INSm{Hamming metric} case, we assume $\cM=\fq^{n\times t}$.

%There exists
%$b \in \fq^n$ such that $V^{(i)}b=0$, $R_ib\neq 0$, so we may write 
%$$ \cZ^{(i)} = \{A + BD_\ga: A \in \cY^{(i)} \cap R_i^\perp, \ga \in \fq^t \backslash \{0\} \},$$
%where $B=[b,...,b]$, $D_\ga=\diag(\ga_1,...,\ga_t)$.

We define the following sets for any non-negative integer $\gd$:
\begin{eqnarray*}
	\cY_\gd^{(i)} &:= & \{A \in \cY^{(i)}: \rk(A)\geq 2\gd+1\}\\
	\cZ_\gd^{(i)}&:=&\{A \in \cZ^{(i)}: \rk(A)\geq 2\gd+1\}
%	\cW_\gd^{(i)}&:=&\{A \in \cW^{(i)} : \rk(A)\geq 2\gd+1\}	
\end{eqnarray*}

Clearly a matrix $L\in \fq^{N \times d_S}$ represents a linear $(\cI,\gd)$-ECIC if and only if for all $i\in[m]$ it holds that
$$
w\left(LV^{(S)}Z\right)\ge 2\gd+1,
$$
for all $Z \in \cZ^{(i)}$ if \INSm{$w$ is the Hamming metric} and for all $Z \in\cZ_\delta^{(i)} \cap \cM^\triangle$ if \INSm{$w$ represents the rank metric}.

\begin{remark}\label{remzdel}
	Note that if $\rk(LV^{(S)}Z)\ge 2\gd+1$ whenever $Z \in \cZ^{(i)}$ has rank at least $2\gd+1$, %and there exists some $Z\in \cZ_\gd^{(i)}$
	then $\rk(LV^{(S)}Z)\ge r$ whenever $Z \in \cZ^{(i)}$ has rank at least $r$ for $r \in [2\gd+1]$.
	This can be seen by the following inductive argument.
	Suppose $\rk(LV^{(S)}Z)\ge r$ whenever $\rk(Z)\geq r$ for some positive integer $r$ and suppose there exists some such $Z \in \cZ^{(i)}$.
	Let $X \in \cZ^{(i)}$ have rank $r-1$ and let $\cS$ be set of $r-1$ linearly independent columns of $X$. 
	If $\rk(LV^{(S)}X) < r-1$ then $S$ cannot be completed to a linearly independent set of size $r$ in $\cZ^{(i)}$ by hypothesis, 
	%for if it could, then we could find
	%$y \in \fq^N$ such that $V^{(i)}y = 0, R_i y = 0$ and $[\cS|y|0] \in \fq^{N \times n}$ has rank $r$, which by assumption yields $\rk(LV^{(S)}[\cS|y|0]) = r$.
	Then every column of $Z$  is contained in the span of $\cS$, so in particular the column space of $X$ contains an $r$-dimensional space, which is impossible. 
\end{remark}  

\INSe{\subsection{Bounds on the Optimal Length of an Error Correcting Index Code}}

We denote by $\cN(\cI,\gd)$ the optimal length $N$ of an \INSe{$\fq$-linear} $(\cI,\gd)$-ECIC. Clearly $\cN(\cI,0)=\gk(\cI)$.
This section is devoted to obtaining bounds on this number. In \cite{Son1} a number of bounds are discussed, namely the
$\ga$-bound, $\gk$-bound and Singleton bound. All of these bounds have extensions for Hamming metric $(\cI,\gd)$-ECICs. The $\ga$-bound holds for rank metric $(\cI,\gd)$-ECICs, but \INSe{the question of the rank distance analogue of the $\gk$-bound still open}.
We consider these two cases separately. 

\subsubsection{\INSe{Hamming metric $(\cI,\gd)$-ECICs}}\label{secham}
%\INSm{\begin{remark}
%We can note that for the Hamming metric case, if L is an $(\cI,\gd)$-ECIC for $t=1$, then L is an $(\cI,\gd)$-ECIC also for $t>1$ and viceversa. 
%Indeed, fix $i\in[m]$, we recall that 
%$$\cZ^(i)=\{Z\in\FF_q^{n\times t}\mid V^{(i)}Z=0\mbox{ and } R_iZ\ne 0\}.$$
%Let?s denote, now, $\cZ^{i}_1$ and $\cZ^{(i)}_t$ the sets $\cZ^{(i)}$'s for the cases $t=1 and $t>1$, respectively. We can note that $Z\in \cZ^{(i)}_t$ is such that
%$$\begin{array}{ccc}Z=[Z_1&\dots&Z_t]\end{array}$$
%with $Z_l\in \cZ^{i}_1\cup \{0\}$ for all $l\in[t]$ and at least one $Z_j\ne 0$.
%Then $w(LV^{(S)}Z)=w(LV^{(S)}Z_1)+...+w(LV^{(S)}Z_t)$. That means $L$ is an $(\cI,\gd)$-ECIC for $t=1$ if and only if L is an $(\cI,\gd)$-ECIC also for $t>1$.
%\end{remark}
%Thanks to the remark above in this section we can consider the case $t=1$ and all the results hold also for $t>1$.
%}

\INSe{We assume throughout this section that $w$ represents the Hamming weight and that $\cN(\cI,\gd)$ is the optimal length of an $\fq$-linear Hamming metric $(\cI,\gd)$-ECIC.}
%We let $E_i \in \fq^{n}$ be the matrix with 1 in the $i$-th position and zeroes elsewhere. Let
%$H \subset [n]$.
We denote by \INSe{$N(k,d)$} the optimal length $\ell$ of an \INSm{$\FF_q$-$[\ell,k,d]$ code, i.e.} a $k$-dimensional $\fq$-linear code in $\fq^{\ell}$ of minimum Hamming distance $d$. 

\INSm{Thanks to the following result we can restrict our study to the case $t=1$.
\begin{lemma}\label{lem:thx}
Let $t\ge1$. Consider two instances $\cI=(t,m,n,\cX,\cX^{(S)},R)$ and $\cI'=(1,m,n,\cX,\cX^{(S)},R)$. Then a matrix $L\in\fq^{N\times d_S}$ represents a linear $(\cI,\gd)$-ECIC if and only if $L$ represents a linear $(\cI',\gd)$-ECIC.
\end{lemma}
\begin{proof}
\INSe{The matrices $V^{(i)}$, $V^{(S)}$ and the request vectors $R_i$'s are the same} for the two instances $\cI$ and $\cI'$. 
Let
$$
\cZ^{(i)}_t=\{Z\in\FF_q^{n\times t}\mid V^{(i)}Z=0\mbox{ and } R_iZ\ne 0\}\mbox{ and }\cZ^{(i)}_1=\{Z\in\FF_q^{n \times 1}\mid V^{(i)}Z=0\mbox{ and } R_iZ\ne 0\}.
$$
If $L$ represents a linear $(\cI,\gd)$-ECIC, then for all $Z\in\cZ^{(i)}_t$ $w(LV^{(S)}Z)\ge 2\gd+1$, where $w$ counts the number of non-zero rows. 
On the other hand, if $L$ realizes a linear $(\cI',\gd)$-ECIC, then for all $Z\in\cZ^{(i)}_1$, $w(LV^{(S)}Z)\ge 2\gd+1$, where in this case the number of non-zero rows is the same of the non-zero entries of the $N \times 1$ vector $LV^{(S)}Z$.\\
Note that any $Z\in\cZ^{(i)}_t$ satisfies 
$$
Z=[Z^1, ... ,Z^t],
$$
with $Z^j\in \cZ^{(i)}_1\cup\{0\}$ for all $j\in [t]$, and at least one is different from zero. 
Without loss of generality, suppose $Z_1\ne 0$. Then if $L$ represents a linear $(\cI',\gd)$-ECIC we have 
$$
LV^{(S)}Z=[LV^{(S)}Z^1,...,LV^{(S)}Z^t]
$$
where the column $LV^{(S)}Z^1$ has at least $2\gd+1$ non-zero entries, which implies that at least $2\gd+1$ rows of $LV^{(S)}Z$ are non-zero.\\
Conversely, let $L\in\fq^{N\times d_S}$ represent a linear $(\cI,\gd)$-ECIC, and let $Z\in\cZ^{(i)}_t$ such that 
$$
Z=[\begin{array}{cccc}
Z^1,0,...,0\end{array}],
$$
with $Z^1\in \cZ^{(i)}_1$. Then 
$$LV^{(S)}Z=[ LV^{(S)}Z^1,0,..., 0]$$ 
has at least $2\gd+1$ non-zero rows, which means that $2\gd+1$ entries of the first column are non-zero. Therefore $L$ represents a linear $(\cI',\gd)$-ECIC.
\end{proof}
}
\INSm{For the remainder of this section, we fix $t=1$, knowing that all the results hold also for $t>1$ from Lemma \ref{lem:thx}.}

We define the set:
%$$\cS(\cI)=\cup_{i \in [m]}  {{\mathcal Z}^{(i)}} \subset \fq^{n}$$
%and
$$\cJ(\cI):=\{U < \fq^n: U \backslash \{0\} \subset \cup_{i \in [m]}\cZ^{(i)} \}. $$
We denote by $\ga(\cI)$ the maximum dimension of any element of $\cJ(\cI)$, that is, the maximum dimension of any subspace of $\fq^{n}$ in $ \cup_{i \in [m]}\cZ^{(i)} \cup \{0\}$. 

%We have the following results.
\INSe{We first give an extended $\ga$-bound, which gives a lower bound on the length of an optimal Hamming metric $(\cI,\gd)$-ECIC. It may be helpful for the reader to think of this as the algebraic analogue of the independence number of a side-information graph.}

\begin{proposition}{($\ga$-bound)}
	Let $\cI$ be an instance of the ICCSI problem. Then
	$$
	N(\ga(\cI),2\gd+1)\le\cN(\cI,\gd).
	$$
\end{proposition}
\begin{proof}
%	Let $L \in \fq^{N\times d_s}$ represent an $\fq$-linear $(\cI,\gd)$-ECIC.	
%	Let $U \in \cJ(\cI)$ have dimension $k$ and let $G$ be a rank $k$ matrix in $\fq^{n \times k}$ such that \INSe{$U = \{GX: X \in \fq^{k\times t}\}\subset \fq^{n \times t}$. Let $C_U =\{ LV^{(S)}G X: X \in \fq^{k\INSe{\times t}} \} \subset \fq^{N \times t}$}. 
%	Then \INSe{every element of }$U \backslash \{0\}$ \INSe{is contained in} some $\cZ^{(i)}$, so $w(LV^{(S)}GX) \geq 2\delta + 1$ for all non-zero $X \in \fq^{k\INSe{\times t}}$, \INSe{by assumption}.
%	\INSe{In particular, this furthermore means that $LV^{(S)}G$ has trivial null-space in $\fq^{k\times t}$, so} that $LV^{(S)}G$ has rank $k$ over $\fq$.
%	It follows that $C_U$ is an $\fq$-$[N,\INSe{kt},2\gd+1]$ code in $\INSe{ \fq^{N\times t}}$ with $N \geq N_{\INSe{t}}(k,2\gd+1)$. 
%	Choosing $U$ of maximal dimension in $\cJ(\cI)$ for a Hamming metric $(\cI,\delta)$-ECIC of optimal length we see that $ N_\INSe{t}(\ga(\cI),2\gd+1)\le\cN(\cI,\gd)$.
Let $L \in \fq^{N\times d_s}$ represent a linear $(\cI,\gd)$-ECIC.	
Let $U \in \cJ(\cI)$ have dimension $k$ and let $G$ be a rank $k$ matrix in $\fq^{n \times k}$ such that $U = \{GX: X \in \fq^k\}.$ 
Let $$C_U =\{ LV^{(S)}G X: X \in \fq^{k} \} \subset \fq^{N }.$$ 
Then \INSe{every element of }$U \backslash \{0\}$ \INSe{is contained in} $\cZ^{(i)}$ for some $i \in [m]$, so $w(LV^{(S)}GX) \geq 2\delta + 1$ for all non-zero $X \in \fq^k$, \INSe{by assumption}.
This furthermore implies that $LV^{(S)}G$ has rank $k$ over $\fq$.
It follows that $C_U$ is an $\fq$-$[N,k,2\gd+1]$ code with $N \geq N(k,2\gd+1)$. 
Choosing $U$ of maximal dimension in $\cJ(\cI)$ for an $(\cI,\delta)-ECIC$ of optimal length we see that $$ N(\ga(\cI),2\gd+1)\le\cN(\cI,\gd).$$
\end{proof}

\INSe{We give sufficient conditions for tightness of the $\ga$-bound.}

\begin{corollary}
	Let $\cI$ be an instance of the ICCSI problem. If there exists a matrix $B \in \fq^{\ga(\cI)\times d_S}$ satisfying
	$B{V^{(S)}}^\perp \cap {V^{(i)}}^\perp \subset {R_i}^\perp $ for all $i \in [m]$ then
	$$
	N(\ga(\cI),2\gd+1) = \cN(\cI,\gd).
	$$ 
\end{corollary}
\begin{proof}
	Let $B \in \fq^{\ga(\cI)\times d_S}$ satisfy the hypothesis of the corollary. Let $G$ be a generator matrix
	for the $\fq$-linear $[N,\ga(\cI),2\gd+1]$ code $\INSe{ \{GX : X \in \fq^{d_S} \}}$, and let $L = GB$. Then 
	$$w(LV^{(S)}Z) = w(GBV^{(S)}Z) \geq 2\gd+1$$ whenever $BV^{(S)}Z \neq 0$. If $BV^{(S)}Z=0$, then $Z \notin \cZ^{(i)}$ 
	for any $i \in [m]$ by our choice of
	$B$, so it follows that $L$ represents an $\fq$-linear Hamming metric $(\cI,\gd)$-ECIC of length 
	$$N=N(\ga(\cI),2\gd+1) \geq  \cN(\cI,\gd).$$ 
\end{proof}

Setting $\delta = 0$ in the above gives the following \INSe{lower bound on the min-rank of an instance as an immediate consequence}.

\begin{corollary}
	Let $\cI$ be an instance of the ICCSI problem. Then
	$$ \ga(\cI) \leq \kappa(\cI),$$
	with equality occurring if there exists $L \in \fq^{\ga(\cI)\times d_S}$ satisfying
	$L{V^{(S)}}^\perp \cap {V^{(i)}}^\perp \subset {R_i}^\perp $ for all $i \in [m]$.
\end{corollary}	 

The reader will observe that in fact the condition $L{V^{(S)}}^\perp \cap {V^{(i)}}^\perp \subset {R_i}^\perp$ for each $i \in [m]$ 
is simply the equivalent statement to that of Lemma \ref{lemdecode}, found by dualization.

Both the $\gk$-bound and the Singleton bound hold in the context of coded-side information. The proofs are trivial extensions of those given in \cite{Son1} and are \INSe{retained here only for the convenience of the reader.}  

\begin{proposition}{($\gk$-bound)}\label{propgk}
	Let $\cI$ be an instance of the ICCSI problem. Then
	$$
	\cN(\cI,\gd)\le N(\gk(\cI),2\gd+1).
	$$
\end{proposition}

%\INSe{\begin{proof}
%		Let $L_1\in\fq^{N_1\times d_S}$ represent an optimal $\fq$-linear $\cI$-IC of length $N_1=\gk(\cI)$.
%		Let $L_2 \in \fq^{N_0\times N_1}$ have rank $N_1$, such that the code $C=\{L_2 X : X \in \fq^{N_1 \times t}\} < \fq^{N_0}$ is a Hamming metric
%		$[N_0,N_1t,2\gd+1]$ $\fq$-linear code in $\fq^{N_0 \times t}$ with $N_0 = N_t(N_1,2\gd+1)$ for some $\gd$.
%		By Corollary \ref{cordec}, $L_1V^{(S)}Z \in \fq^{N_1\times t}$ is non-zero for all $Z$ in any $\cZ^{(i)}$,  $w(L_2L_1V^{(S)} Z)\geq 2 \delta + 1$ for all such $Z$. 
%		Then $L=L_2L_1$ represents an $\fq$-linear Hamming metric $(\cI,\gd)$-ECIC of length 
%		$N_0=N_t(\gk(\cI),2\gd+1) \geq \cN(\cI,\gd)$.
\INSe{\begin{proof}
Let $L_1\in\fq^{N_1\times d_S}$ represent an optimal $\fq$-linear $\cI$-IC of length $N_1=\gk(\cI)$.
Let $L_2 \in \fq^{N\times N_1}$ have rank $N_1$, such that the code $$C=\{L_2 X : X \in \fq^{N_1}\} < \fq^N$$ is an
$[N,N_1,2\gd+1]$ linear code over $\fq$ with $N = N(N_1,2\gd+1)$ for some $\gd$.
Since $LV^{(S)}Z \in \fq^{N_1}$ is non-zero for all $Z \in \cS(\cI)$,  $$w(L_2L_1V^{(S)} Z)\geq 2 \delta + 1,$$ for all such $Z$. 
Then $L=L_2L_1$ represents a linear $(\cI,\gd)$-ECIC of length 
$$N=N(\gk(\cI),2\gd+1) \geq \cN(\cI,\gd).$$
	\end{proof}}
	
	\begin{proposition}{(Singleton bound)}
		Let $\cI$ be an instance of the ICCSI problem. Then
		$$ \gk(\cI)+2\gd \leq\cN(\cI,\gd).$$
	\end{proposition}
	
%	\INSe{\begin{proof}
%			Let $L\in\fq^{N\times d_S}$ represent an optimal linear Hamming metric $(\cI,\gd)$-ECIC over $\fq$, so that $N=\cN(\cI,\gd)$. 
%			Let $L'$ the matrix obtained by deleting any $2\gd$ rows of $L$. 
%			By Theorem \ref{th:err}, for each $i\in[m]$, $w\left(LV^{(S)}Z\right)\ge 2\gd+1$ for all $Z \in {\mathcal Z}^{(i)}$, so that 
%			$w\left(L'V^{(S)}Z\right)\ge 1,$
%			for all such $Z$. So $L'$ is a linear index code of length $N-2\delta$ for the instance $\cI$. Now $L'$ has at least $\gk(\cI)$ rows so that
%			$\gk(\cI) \leq \cN(\cI,\gd)-2\gd$.
\INSe{\begin{proof}
	Let $L\in\fq^{N\times d_S}$ represent an optimal linear $(\cI,\gd)$-ECIC over $\fq$, so that $N=\cN(\cI,\gd)$. 
	Let $L'$ the matrix obtained by deleting any $2\gd$ rows of $L$. 
	By Theorem \ref{th:err}, for each $i\in[m]$, 
	$$w\left(LV^{(S)}Z\right)\ge 2\gd+1,\; \text{ for all }Z \in {\mathcal Z}^{(i)},$$
	so that 
	$$w\left(L'V^{(S)}Z\right)\ge 1,\; \text{ for all }Z \in {\mathcal Z}^{(i)}.$$
	So $L'$ is a linear index code of length $N-2\delta$ for the instance $\cI$. Now $L'$ has at least $\gk(\cI)$ rows so that
	$$\gk(\cI) \leq \cN(\cI,\gd)-2\gd.$$
		\end{proof}}
		
		%%Recall the well-known Singleton bound (cf. \cite{HP03}), which states that $k<N-2\gd$ for any $\fq-[N,k,2\delta+1]$ code.
		
		\begin{example}
			Let $m=6,n=5,q=2$ and let $\cX^{(S)}=\fq^5$.
			Suppose that $\cX^{(i)}$ has dimension $d_i=2$ for each $i \in \{1,...,6\}$.
			Let $\cI$ be the instance defined by 
			user side-information
			$$
			V^{(1)} = \left[\begin{array}{ccccc} 
			0 & 1 & 1 & 1 & 0 \\
			0 & 0 & 1 & 1 & 1
			\end{array}\right],
			V^{(2)} = \left[\begin{array}{ccccc} 
			1 & 0 & 0 & 0 & 1\\
			0 & 0 & 1 & 1 & 0
			\end{array}\right],
			V^{(3)} = \left[\begin{array}{ccccc} 
			1 & 1 & 1 & 1 & 0 \\
			0 & 0 & 0 & 1 & 1
			\end{array}\right],
			$$
			$$         
			V^{(4)} = \left[\begin{array}{ccccc} 
			1 & 0 & 0 & 1 & 0 \\
			0 & 1 & 1 & 1 & 1
			\end{array}\right],
			V^{(5)} = \left[\begin{array}{ccccc} 
			0 & 0 & 1 & 1 & 0 \\
			0 & 0 & 0 & 1 & 1
			\end{array}\right],
			V^{(6)} = \left[\begin{array}{ccccc} 
			1 & 0 & 0 & 1 & 0 \\
			0 & 0 & 1 & 1 & 0
			\end{array}\right], 
			$$
			and requests 
			$$R_1=[10000],R_2=[10000],R_3=[00101],R_4=[10001],R_5=[11000],R_6=[00111].$$
			It can be checked that $\gk(\cI)=3$. 
			Moreover, $\cup_{i \in [6]} \cZ^{(i)} $ contains the non-zero elements in the span of
			$\{		(1 0 1 1 0),	(1 0 0 1 1),	(0 1 1 0 0)	\}$. Then $\alpha(\cI)= \gk(\cI)=3.$
			It follows from the $\ga$-bound that $6=N(3,3)=N(\ga(\cI),3) \le \cN(\cI,1)$.
			From the $\gk$-bound we have $6=N(3,3)= N(\gk(\cI),3) \ge \cN(\cI,1)$.
		\end{example}
		
		Reed-Solomon codes and their $q$-analogues in the form of Gabidulin codes \cite{G85} are examples of MDS and {\em maximum rank distance} (MRD) codes respectively. They were first introduced in \cite{D78}. In fact any extended generalized Reed-Solomon code over $\fq$ is an MDS code of length $q+1$ \cite[Theorem 5.3.4]{HP03} so the existence of such codes is assured for such lengths. It is conjectured that any $\fq$-$[N,k,d]$ MDS code satisfies $N \leq q + 1$ unless $q$ is even and $k=3$ or $k=q-1$ (in which case $N \leq q+2$) \cite{HP03}.

		\begin{corollary}
			Let $\cI$ be an instance of the ICCSI problem. If $q\ge \gk(\cI)+2\gd-1$ then
			$$
			\cN(\cI,\gd)= \gk(\cI)+2\gd.
			$$
		\end{corollary}
		
		\begin{proof}
			If $q\ge \gk(\cI)+2\gd-1$ then there exists an $\fq$-linear $[q+1,\gk(\cI),2\delta+1]$ MDS code, namely an extended Reed-Solomon code. Then we obtain 
			$$ \gk(\cI)+2\gd \leq \cN(\cI,\gd) \leq N(k(\cI),2\delta + 1) = \gk(\cI)+2\gd.$$ 
		\end{proof}
		
		As usual, we let $V_q(n,r)$ be the size of a Hamming sphere of radius $r$ in $\fq^n$. We have the following generalization of \cite[Theorem 6.1]{Son1}. 
		
		\begin{theorem}\label{th:random2}
			Let $\cI$ be an instance of the ICCSI problem. Let $L\in\fq^{N\times d_S}$ be selected uniformly at random over $\fq$. 
			The probability that $L$ corresponds to a Hamming metric $\fq$-linear $(\cI,\gd)$-ECIC is at least 
%			$$
%			\INSe{1-\sum_{i=1}^mq^{t(n-d_i-1)}(q^t-1)\frac{V_q(N,2\gd)}{q^N}}.
%			$$
%			In particular there exists an $\fq$-linear Hamming metric $(\cI,\gd)$-ECIC if 
%			$$\INSe{N > t(n-d-1) + \log_q(m(q^t-1)V_q(N,2\gd)),}$$
			$$
			1-\sum_{i=1}^mq^{n-d_i-1}(q-1)\frac{V_q(N,2\gd)}{q^N}.
			$$
			In particular there exists an $\fq$-linear $(\cI,\gd)$-ECIC \INSe{of length $N$} if 
			$$N > n-d-1 + \log_q(m(q-1)V_q(N,2\gd)),$$
			where $d = \min\{d_i: i \in [m]\}$.
		\end{theorem}
		\begin{proof}
%			Let $L$ be selected uniformly at random in $\fq^{N\times d_S}$.
%			If $w(LV^{(S)}Z)\leq 2\gd$ for some $Z$ in $\cZ^{(i)}$ then $L$ is not $\gd$-delta error correcting at the $i$th decoder.
%			The probability of this occurring at the $i$th receiver is upper bounded by
%			$$\frac{|\cZ^{(i)}|V_q(N,2\gd)}{q^{N}} =\INSe{\frac{q^{t(n-d_i-1)}(q^t-1)}{q^N} V_q(N,2\gd)=q^{t(n-d_i-1)-N}(q^t-1) V_q(N,2\gd)}$$
%			so from the union bound the probability of this occurring at some $i$th decoder is at most
%			$$\INSe{\sum_{i \in [m]} q^{t(n-d_i-1)-N}(q^t-1) V_q(N,2\gd) \leq mq^{t(n-d-1)}(q^t-1)\frac{ V_q(N,2\gd)}{q^N}}.$$
Let $L$ be selected uniformly at random in $\fq^{N\times d_S}$.
	If $w(LV^{(S)}Z)\leq 2\gd$ for some $Z$ in $\cZ^{(i)}$ then $L$ is not $\gd$-delta error correcting at the $i$th decoder.
	The probability of this occurring at the $i$th receiver is upper bounded by
	$$\frac{|\cZ^{(i)}|V_q(N,2\gd)}{q^{N}} =q^{n-d_i-1-N}(q-1) V_q(N,2\gd)$$
	so from the union bound the probability of this occurring at some $i$th decoder is at most
	$$\sum_{i \in [m]} q^{n-d_i-1-N}(q-1) V_q(N,2\gd) \leq mq^{n-d-1}(q-1)\frac{ V_q(N,2\gd)}{q^N}.$$
		\end{proof}
		
		\begin{remark}
			\INSe{If we let $m''$ denote the number of equivalence classes of $[m]$ under the relation $\hat{m}$ that $i$ and $j$
			are equivalent if $\cZ^{(i)}=\cZ^{(j)}$, then in the above we obtain the following refinement:
			There exists an $\fq$-linear $(\cI,\gd)$-ECIC \INSe{of length $N$} if
			$$N > n-d-1 + \log_q(m''(q-1)V_q(N,2\gd)),$$
			where $d = \min\{d_i: i \in [m]\}$.}
		\end{remark}
		
		%In the restriction to the ICSI problem, this bound is tighter than that of \cite[Theorem 6.1]{Son1}, which it generalizes.
		
		Let $H_q$ denote the $q$-ary entropy function:
		$$
		H_q:(0,1)\to \mathbb R:x\mapsto x\log_q(q-1)-x\log_q(x)-(1-x)\log_q(1-x).
		$$
		It is well known that the function $H_q(x)$ is continuous and increasing on $(0,1-(1/q))$. 
		A proof of the following can be found in \cite{L94}.
		\begin{lemma}\label{lm:hq} Let $\lambda \in (0,1-(1/q))$ be such that $n\lambda$ is an integer. Then
			$$
			V_q(n,\lambda n)\le q^{H_q(\lambda)n}.
			$$	
		\end{lemma}
		
		\begin{corollary}
			Let $\cI$ be an instance of the ICCSI problem with. Let $\gl\in {\mathbb Q}$ such that $0<\gl<1-1/q$ and let $N \in {\mathbb Z}$ satisfy 
			$\gl N \in {\mathbb Z}$ %and %$$N-H_q(\gl)N>\log_q\left(\sum_{i=1}^mq^{n-d_i-1}\right).$$ 
			Then, choosing the entries of $L\in\fq^{N\times d_S}$ uniformly at random over the field $\fq$, the probability that $L$ corresponds to a \INSe{Hamming metric $\fq$-linear} $(\cI,\gd)$-ECIC, with 
			$\gd=\left\lfloor\frac{\gl N}2\right\rfloor$, is at least 
			$$
			{1- (q-1)\sum_{\INSe{i \in \hat{m}}}\frac{q^{(n-d_i-1)}}{q^{N(1-H_q(\gl))}}}.
			$$
			In particular there exists an $\fq$-linear Hamming metric $(\cI,\gd)$-ECIC if 
			$${\INSe{m''} < \frac{q^{N(1-H_q(\gl))-(n-d-1)}}{q-1}},$$
			where $d = \min\{d_i: i \in [m]\}$.
		\end{corollary}

		%%%%%%%%%%%%%%%%%%%%%%%%%%random to have a bound
		
\subsubsection{\INSe{Rank Metric $(\cI,\gd)$-ECICs}}\label{sec:random}

\INSe{We assume throughout this section that $t>1$, that $w$ represents the rank weight and that $\cN(\cI,\gd)$ is the optimal length of an $\fq$-linear rank metric $(\cI,\gd)$-ECIC. Again, we fix some further notation.}
We let $N(t,\log_q M,d)$ denote the least integer $s$ such that there exists a code in $\fq^{s \times t}$ of minimum rank distance $d$
and size $M$. \INSe{We say that an $\fq$-linear code of dimension $k$ and minimum rank distance $d$ in $\fq^{n \times t}$ is a rank metric $\fq$-$[n,k,d]$ code.}  
In analogy with the previous section, we define the set:
$$\cJ(\cI):=\{ U \subset \fq^{n \times t}: X-X' \in \cup_{i \in [m]} \cZ_\gd^{(i)},\text{for any } X, X' \in U\}.$$
and let $\ga(\cI):=\max\{ \log_q|U|: U \in \cJ(\cI) \}$.

\begin{theorem}
	Let $\cI$ be an instance of the ICCSI problem. Then
	$$ N(t,\ga(\cI),2\gd+1) \leq \cN(\cI,\gd).$$
\end{theorem}

\begin{proof}
	Let $L \in \fq^{N \times t}$ represent an optimal $(\cI,\gd)$-ECIC. 
	Let $U\in \cJ(\cI)$ and define $$C_U=\{LV^{(S)}X : X \in U \} \subset \fq^{N \times t}.$$
	Then $C_U$ has minimum rank distance $2\gd+1$ in $\fq^{N\times t}$ since $w(LV^{(S)}(X-X'))\geq 2\gd+1$
	for any pair $X,X' \in U$. The result follows on choosing $U \in \cJ(\cI)$ such that $\log_q|U| = \ga(\cI)$.
\end{proof}

The rank-distance Singleton bound \cite{D78} states that for any code $C$ in $\fq^{N\times t}$ of minimum distance $2\gd+1$ 
that 
$$ \log_q|C| \leq \left\{ \begin{array}{ll}
t(N-2\delta) & \text{ if } t \geq N\\ 
N(t-2\delta) & \text{ if } t \leq N \\
\end{array}        
\right.$$
Codes that meet this bound are called maximum rank distance (MRD) codes. 
Combining the $\ga$-bound and the Singleton bound for rank-metric codes immediately yields the following.

\begin{corollary}(Singleton bound)\label{coracor}
	Let $\cI$ be an instance of the ICCSI problem. Then
	$$\cN(\cI,\gd)\geq \left\{
	\begin{array}{ll}
	\displaystyle{\frac{\alpha(\cI)}{t} +2\delta}& \text{ if } t \geq N(t,\ga(\cI),2\gd+1), \\
	\displaystyle{\frac{\ga(\cI)}{t-2\delta}} & \text{ if }t \leq N(t,\ga(\cI),2\gd+1).
	\end{array}\right.    
	$$
\end{corollary}

We now give a result on the existence of a linear encoding of length $N$ for $(\cI,\delta)$-ECIC, extending Theorem $6.1$ in \cite{Son1}. 
We let $V_q(N,t,s):= | \{ X \in \fq^{N\times t} : w(X) \leq s \} |$ denote the size of a sphere of rank distance radius $s$ in $\fq^{N \times t}$.

We will use the following result from \cite{LT}. 
%For positive integers $r,q,t$ define
%$$ Q_r(q^t): = (q^t-1)(q^t-q)\cdots(q^t-q^{r-1}),$$
%if $r\leq t$ and set $Q_r(q^t)=0$ otherwise.

\begin{theorem}\label{thlt}
	Let $W$ be an $\fq$-vector space and let $\cF_r$ be a family of $r$-dimensional subspaces of $W$. Let $\ho_{\cF_r} (\FF_q^t,W)$ denote the set of 
	homomorphisms of $\fq^t$ whose images lie in $\cF_r$. Then
	$$|\ho_{\cF_r}(\FF_q^t,W)|=|\cF_r|\prod_{j=0}^{r-1} (q^t-q^j). $$ 
\end{theorem}

If $W$ is a subspace of $\fq^s$ then $\ho_{\cF_r}({\fq}^t,W)$ corresponds to the set of all $s \times t$ matrices of rank $r$ whose column spaces lie in $W$. 

Let $m'''$ denote the number of equivalence classes of $[m]$ under the relation $\breve{m}$ that $i$ and $j$
are equivalent if $\cZ^{(i)}_\gd=\cZ^{(j)}_\gd$

\begin{theorem}\label{thrandom}
	Let $\cI$ be of an instance of an ICCSI problem and let $L \in \fq^{N \times d_S}$ for some positive integer $N$. 
	The probability that $L$ represents a linear $(\cI,\delta)$-ECIC of length $N$ is at least
	\begin{eqnarray*}
		1- q^{-N t} \sum_{i \in \INSe{\breve{m}}}  \sum_{r\geq 2\gd +1}\left(\prod_{j=0}^{r-1}( q^{n-d_i}-q^j)-\prod_{j=0}^{r-1} (q^{n-d_i-1}-q^j)\right)\qbin{t}{r}{q} 
		\sum_{r=0}^{2\gd} \prod_{j=0}^{r-1} (q^{N}-q^j) \qbin{t}{r}{q}.
	\end{eqnarray*}
	In particular, there exists such a matrix $L$ if
	$$\sum_{i \in \INSe{\breve{m}}} \sum_{r\geq 2\gd +1}\left(\prod_{j=0}^{r-1} (q^{n-d_i}-q^j)-\prod_{j=0}^{r-1}( q^{n-d_i-1}-q^j)\right)
	\qbin{t}{r}{q} <  \frac{q^{Nt}}{V(N,t,2\delta)} .$$
\end{theorem}

\begin{proof}
	From Theorem \ref{th:err}, the matrix $L \in \fq^{N \times d_S}$ represents a linear $(\cI,\delta)$ if and only if for each $i \in [m]$, 
	$w(LV^{(S)}Z) \geq 2 \delta + 1$ for any $Z \in {\cZ}^{(i)}_\gd$.
	Therefore, a decoding failure at the $i$th node occurs if and only if the sphere $B_{2\gd}(Z)=\{LV^{(S)}Z + W : w(W)\leq 2\gd\}\subset \fq^{N \times t}$
	contains the zero matrix for some $Z \in \cZ_\delta^{(i)}$.
	Then the probability of a decoding failure at the $i$th receiver is upper bounded by
	$$\frac{\left|\cup_{Z \in \cZ^{(i)}_\gd } B_{2\gd}(Z)\right|}{|\fq^{N \times t}|} \leq  \frac{|\cZ^{(i)}_\gd| V(N,t,2\gd) }{q^{Nt}}.$$ 
	We define the following sets for each non-negative integer $r$:
	$$S^{(i)}_r = \{ M < {V^{(i)}}^\perp:\dim\; M =r \} \text{ and } T^{(i)}_r = \{ M < {V^{(i)}}^\perp \cap {R_i}^\perp:\dim\; M =r \}. $$
	Then $|S^{(i)}_r| =  \qbin{n-d_i}{r}{q}$ and $|T^{(i)}_r| = \qbin{n-d_i-1}{r}{q}.$
	%From Theorem \ref{thlt}, it follows that 
	%$$|\ho_{S^{(i)}_r}({\fq}^t,{V^{(i)}}^\perp)|=|S^{(i)}_r|Q_r(q^t)  \text{ and } |\ho_{T^{(i)}_r}({\fq}^t,{V^{(i)}}^\perp \cap {R_i}^\perp)|=|T^{(i)}_r|Q_r(q^t). $$ 
	Now $S^{(i)}_r$ (resp. $T^{(i)}_r$) is the set of column spaces in $\fq^n$ of all $n \times t$ matrices in $\cY^{(i)}$ (resp. in $\cW^{(i)}$) of rank $r$. %, so that
	%$$\ho_{S^{(i)}_r}({\fq}^t,{V^{(i)}}^\perp) = \cY^{(i)} \text{ and }\ho_{T^{(i)}_r}({\fq}^t,{V^{(i)}}^\perp \cap {R_i}^\perp) = \cW^{(i)}$$ 
    Then from Theorem \ref{thlt}, it follows that
	\begin{eqnarray*}
		|\cZ^{(i)}_\gd| & = & \sum_{r\geq 2\gd +1}\left( |\ho_{S^{(i)}_r}(\fq^t,{V^{(i)}}^\perp)|-|\ho_{T^{(i)}_r}(\fq^t,{V^{(i)}}^\perp \cap {R_i}^\perp)| \right)\\
		& = & \sum_{r\geq 2\gd +1}\left( |S^{(i)}_r|-|T^{(i)}_r|  \right)\prod_{j=0}^{r-1} (q^{t}-q^j)\\
		& = & \sum_{r\geq 2\gd +1}\left(\prod_{j=0}^{r-1} (q^{n-d_i}-q^j)-\prod_{j=0}^{r-1}( q^{n-d_i-1}-q^j)\right)\qbin{t}{r}{q}. 
	\end{eqnarray*}
	Theorem \ref{thlt} can also be applied to obtain
	$$ V(N,t,2\gd) = \sum_{r=0}^{2\gd} |\ho_r(\fq^t,\fq^N)| =  \sum_{r=0}^{2\gd} \prod_{j=0}^{r-1} (q^{N}-q^j) \qbin{t}{r}{q}.$$
	Then the probability of a failure at the $i$th decoder is upper-bounded by
	$$ q^{-Nt}\sum_{r\geq 2\gd +1}\left(\prod_{j=0}^{r-1} (q^{n-d_i}-q^j)-\prod_{j=0}^{r-1}( q^{n-d_i-1}-q^j)\right)\qbin{t}{r}{q} \sum_{r=0}^{2\gd} \prod_{j=0}^{r-1} (q^{N}-q^j) \qbin{t}{r}{q}. $$
	
	The result now follows from the union bound.    
\end{proof}

In the error-free case, that is for $\gd=0$, Theorem \ref{thrandom} asserts that there exists an $N \times d_S$ matrix $L$ of rank $N$ representing an $\fq$-linear $\cI$-IC whenever
$$ 1 > \sum_{i \in \INSe{\breve{m}}}\frac{|\cZ^{(i)}|}{q^{Nt}}=\sum_{i \in \INSe{\breve{m}}} q^{(n-d_i-N-1)t}(q^t-1).$$
\INSm{Moreover $$m'''q^{(k-N-1)t}(q^t-1)\ge \sum_{i \in \INSe{\breve{m}}}^m q^{(n-d_i-N-1)t}(q^t-1)$$}
where $k = \INSe{\max}\{n-d_i:i \in [m]\}$. Then for $N=k+\ell$, there exists a linear $\cI$ of length $N$ as long as
$\INSe{m'''} \leq q^{(\ell+1) t}/(q^t-1).$ In particular, this shows that:

\begin{corollary}
	Let $\cI$ be an instance of the of the ICCSI problem and let  $k = \INSe{\max}\{n-d_i:i \in [m]\}$. If $\INSe{m'''} \leq q^{\ell t}/(q^t-1)$ 
	then $\gk(\cI) \leq k+\ell-1$. 
\end{corollary}	

In \INSe{Table IV} we give parameters $t,N,\delta$ for which the existence if a linear $(\cI,\gd)$-ECIC of length $N$ is established by Theorem \ref{thrandom}
for $n=20$, $d_i=d$ for each $i \in [m]$.
\begin{table}	
	\begin{center}
		\begin{tabular}{|c|c|c|c|c|}
			\hline
			$t \geq $  & $ d \geq $ &  $N \geq$ & $\gd \leq $  & $m \leq $\\
			\hline
			6          & 11         &  16       & 1            & 239 \\ 
			%6              & 11         &  17       & 1       & $16^5$ \\         
		    %6              & 12         &  15       & 1       & $16^7$ \\ 
			7          & 11         &  15       & 1            & 239\\
			%7              & 12         &  14       & 1 & 61200  \\
			%7          & 12         &  15       & 1& 64172855503\\
			%8           & 11         &  15       & 1 & 15667200 \\
			%8         & 12   &  15       & 1 & 67290111901900802 \\
			8          & 12         &  13       & 1            & 239\\
			11         & 12         &  12       & 1            & 239  \\
			20         & 12         &  11       & 1            & 239 \\
			%16             & 12         &  15       & 2  & 61184 \\
			\hline
			10         & 11         & 20        & 2            & 239 \\
			11         & 11         & 19        & 2            & 239 \\
			12         & 11         & 18        & 2            & 239 \\		
			14         & 11         & 17        & 2            & 239 \\
			17         & 11         & 16        & 2            & 239 \\
			9          & 12         & 19        & 2            & 239 \\
			10         & 12         & 18        & 2            & 239 \\
			11         & 12         & 17        & 2            & 239 \\
			13         & 12         & 16        & 2            & 239 \\	
			16         & 12         & 15        & 2            & 239 \\	
			\hline
			16         & 12         & 19        & 3            & 239\\
			19         & 12         & 18        & 3            & 239\\
			13         & 13         & 20        & 3            & 239 \\
			14         & 13         & 19        & 3            & 239 \\
			15         & 13& 18& 3& 239 \\
			16         & 13& 18& 3& 239 \\
			17         & 13& 17& 3& 239 \\
			18         & 13& 18& 3& 239 \\
			19         & 13& 17& 3& 239 \\
			\hline   
			    
		\end{tabular}
		\caption{}
	\end{center}
\end{table}

%%%%%%%%%%%%%%%%%%%%%%%%%%%%%%%%%%%%%%%%

\section{Decoding Index Codes}

\INSe{Error correction for index codes (as for non-multicast network codes) is non-trivial. We consider two approaches, one for rank-metric error correction and the other to correct Hamming metric errors, based on syndrome decoding.}

\subsection{Syndrome Decoding for Hamming Metric Errors}
In \cite{Son1} the authors give a syndrome decoding algorithm for \INSe{Hamming metric error correction in} the ICSI problem. In this section we extend the algorithm to the case of ICCSI problem.

For the remainder \INSe{of this section}, we let $L\in\fq^{N\times d_S}$ be a matrix corresponding to an $(\cI,\gd)$-ECIC. Suppose that for some $i\in[m]$ the $i$th user, receives the message
$$
{Y}_{(i)}=LV^{(S)}{X}+W_{(i)} \;\INSe{\in \fq^{N \times t}},
$$
where $LV^{(S)}{X}$ is the codeword transmitted by $S$ and $W_{(i)}$ is the error vector in $\fq^N$.
Since $R_i \notin \cX^{(i)}$, there exists an invertible matrix $M_{(i)}\in\FF_q^{n\times n}$ such that 
$$
V^{(i)}M_{(i)}=[I|0]\mbox{ where $I$ is the identity matrix in $\FF^{d_i\times d_i}$, and }{R}_iM_{(i)}=\be_{d_i+1}.
$$
%for some $f(i) \notin [d_i]$. 
The matrix $M_{(i)}$ may be constructed by the $i$th user as follows.
Choose a right inverse $A_{(i)} \in \fq^{n \times d_i+1}$ of the matrix $G \in \fq^{d_i+1 \times n}$ that has the rows of $V^{(i)}$ as its first $d_i$ rows and has $R_i$ in the final row.
%\left[\begin{array}{l} V^{(i)}\\R_i \end{array}\right]$. 
Then
%$$\left[\begin{array}{l} V^{(i)}\\R_i \end{array}\right] A_{(i)} = \left[\begin{array}{ll} V^{(i)}A_{(i)}^{[d_i]} & V^{(i)} A_{(i)}^{1} \\
%                                                                                             R_i  A_{(i)}^{[d_i]}& R_i A_{(i)}^{1} \end{array}\right] = I_{d_i+1}. $$
$G A_{(i)} $ is the identity matrix in $\fq^{(d_i+1)\times (d_i+1)}$ so that 
$$V^{(i)}A_{(i)}^{d_i+1} = 0, V^{(i)}A_{(i)}^{[d_i]} = I, R_i  A_{(i)} = [0,...,0,1].$$
Choose $B_{(i)}$ to be an $n \times (n-d_i-1)$ matrix whose columns form a basis of $G^\perp$. Then
$M_{(i)} = [A_{(i)} | B_{(i)}]$ is invertible and satisfies $\INSe{V^{(i)}M_{(i)}}  = [I |0]$ and $R_i M_{(i)}={\bf e}_{d_i+1}$. 
 
Now define $\INSe{X':=M_{(i)}^{-1}X \in \fq^{n \times t}}$. Then we have
$$
V^{(i)}_jX={\bf e}_jM_{(i)}^{-1}X=X_j'\mbox{ for $j \in[d_i]$}
$$
and 
$$
R_iX={\bf e}_{d_i+1}M_{(i)}^{-1}X=X_{d_{i+1}}'.
$$

\begin{lemma}\label{lm:syn}
	%Let $E=\langle [I|0]\rangle < \fq^n$. Then 
	If $Z\in [I|0]^\perp$ and $Z_{d_i+1}\ne 0$ then
	$$
	w(LV^{(S)}M_{(i)}Z)\ge 2\gd +1.
	$$
\end{lemma}
\begin{proof}
	Let $Z\in [I|0]^\perp$ be such that $Z_{d_i+1}\ne 0$.
	Then $V^{(i)}M_{(i)} Z = [I|0] Z = 0$ and $R_iM_{(i)}Z = \be_{d_i+1}Z=Z_{d_i+1}\ne 0$ so $M_{(i)}Z \in \cZ^{(i)}$. The result now follows from Theorem \ref{th:err}.
\end{proof}

%Define, now, the sets
%$$
%\cX_i'=\{1,\dots,,d_i\}\mbox{ and }\cY_i'=[n]\setminus\cX_i'\cap\{f(i)\},
%$$
Let $L'=LV^{(S)}M_{(i)}$ and let $\overline{[s]}:=[n]\setminus [s]$. Consider the following two codes.
We define $\cC^{(i)}\INSe{\subset \fq^N}$ to be the column space of the matrix $[L'^{d_i+1} | L'^{\overline{[d_i+1]}}] \;\INSe{\in \fq^{N \times n}}$ and we define
$\cC_{(i)} \INSe{\subset \fq^N} $ to be the subspace of $\cC^{(i)}$ spanned by the columns of $L'^{\overline{[d_i+1]}}$.
%$$
%\cC^{(i)}:=Span_q(\{L'^{d_i+1}\}\cup\{L'^j\}_{j\in\overline{[d_i]}})
%\text{ and }
%\cC_{(i)}:=Span_q(\{L'^j\}_{j\in\overline{[d_i]}}}).
%$$

For each $i\in[m]$, we have $\cC_{(i)}\subseteq\cC^{(i)} $ with $\dim (\cC^{(i)})=\dim (\cC_{(i)}) +1$. \INSe{As usual, for an $\fq$-linear code $C\in \fq^N$ we write 
	$C^\perp:=\{ y \in \fq^N : x \cdot y = 0 \}$ to denote its dual code. Then we have} ${\cC^{(i)}} ^ \perp \subseteq{\cC_{(i)}}^\perp$ with $\INSe{r_{i}=}\;\dim( {\cC_{(i)}}^\perp)=\dim ({\cC^{(i)}}^\perp) +1$ \INSe{for some $r_{i}$. Let } $H_{(i)}$ be a parity check matrix of $\cC_{(i)}$ of the form

\begin{gather}\label{eq:hi}
H_{(i)}=\left[ \begin{array}{r}
h_{(i)}\\
\hline
H^{(i)}
\end{array}\right] \in \fq^{r_i \times N},
\end{gather}

where $H^{(i)}$ is a parity check matrix of ${\cC^{(i)}}$ and $h_{(i)}\in {\cC_{(i)}} ^ \perp \setminus{\cC^{(i)}}^\perp$.

\INSe{Then}
$$
H_{(i)}L'^{d_i+1}=[s_{d_i+1},0,\dots,0]^T
$$
for some $s_{d_i+1}\in{\fq}\setminus\{0\}$.

\INSe{We now outline a procedure for decoding the demand $R_iX$ at the $i$th receiver, which is based on syndrome decoding. In the first step we compute syndrome, of $H_{(i)}$, in which is embedded a syndrome of $H^{(i)}$. In the second step a table of syndromes is computed for the code $\cC^{(i)}$. Finally, in the third step the output $R_iX$ is computed.}

\begin{itemize}
	
	\item[{\bf Step $I$:}] Compute
	
	\begin{equation}
	H_{(i)}(Y_{(i)}-L'^{[d_i]}X'_{[d_i]})=\INSe{\left[ \begin{array}{c} \alpha_i \\ \gb_i \end{array} \right]} \; \INSe{\in \fq^{r_{i}  \times t}}
	\end{equation}
	
	%where $x_{f(i)}s_{f(i)}+h_{(i)}\cdotW_{(i)}=\ga_i$ and $\gb_i=H^{(i)}W_{(i)}$.
	
	\item[{\bf Step $II$:}] Find $\gep\; \INSe{\in \fq^{N \times t}}$ with $w(\gep)\le\gd$ such that
	
	\begin{equation}\label{eq:eps}
	H^{(i)}\gep=\gb_i\;\INSe{\in \fq^{(r_{i}-1)  \times t}} .
	\end{equation}
	
	\item[{\bf Step $III$:}] Compute
	
	\begin{equation}
	\hat{X}_{d_i+1}=(\ga_i-\INSe{h_{(i)}\gep}) / s_{d_i+1}.
	\end{equation}
\end{itemize}

\begin{theorem}
	If $w(W_{(i)})\le \gd$ then the procedure above has output $\INSe{\hat{X}_{d_i+1}}=X'_{d_i+1}=R_iX$. 
\end{theorem}
\begin{proof}
	We have
	
	\begin{eqnarray*}
	Y_{\INSe{(i)}}&=&LV^{(S)}X+W_{(i)}\\ &= & LV^{(S)}M_{(i)}M_{(i)}^{-1}X+W_{(i)}\\
	&=&L'X'+W_{(i)}\\ & =& L'^{[d_i]}X'_{[d_i]}+L'^{{\overline{[d_i+1]}}}X'_{{\overline{[d_i+1]}}}+L'^{d_i+1}X'_{d_i+1}+W_{(i)},
	\end{eqnarray*}
	
	and $$Y_{(i)}-L'^{[d_i]}X'_{[d_i]}=L'^{{\overline{[d_i+1]}}}X'_{{\overline{[d_i+1]}}}+L'^{d_i+1}X'_{d_i+1}+W_{(i)}.$$
	Then
	$$
	H_{(i)}(Y_{\INSe{(i)}}-L'^{[d_i]}X'_{[d_i]})=\INSe{\left[ \begin{array}{c} \alpha_i \\ \gb_i \end{array} \right]}
	$$
	where $\ga_i=s_{d_i+1} X_{d_i+1}+h_{(i)} W_{(i)}\; \INSe{\in \fq^{1 \times t}}$ and $\gb_i=H^{(i)} W_{(i)}\;\INSe{\in \fq^{(r_i-1) \times t}}$.
	
	Let $\gep \;\INSe{\in \fq^{N \times t}}$ have Hamming weight at most $\gd$, as in Step II. Then $W_{(i)}-\gep\in\cC^{(i)}$ and $w(W_{(i)}-\gep)\le 2\gd$, \INSe{which means} $W_{(i)}-\gep\in\cC_{(i)}$ \INSe{from} Lemma \ref{lm:syn}.
	\INSe{Therefore},
	$$
	(\ga_i-h_{(i)}\gep)/s_{d_i+1}=(s_{d_i+1} X'_{d_i+1}+ h_{(i)}(W_{(i)}-\gep))/ s_{d_i+1}=X'_{d_i+1}.
	$$
\end{proof}

\begin{example}
{ 
Let $q=2$, $m=n=4$, $t=1$, $\gd=1$ and $R_i={\bf e}_i$ for all $i \in [4]$ and $\cX^{(S)}=\FF_2^4$. 

Assume

$$
V^{(1)}=\left[\begin{array}{cccc}
	0&1&1&0\\
	0&0&1&0\end{array}\right], \, V^{(2)}=\left[\begin{array}{cccc}
	1&0&0&0\\
	0&0&1&1\end{array}\right],
$$

$$
V^{(3)}=\left[\begin{array}{cccc}
	1&0&0&0\\
	0&0&0&1\end{array}\right],\,
V^{(4)}=\left[\begin{array}{cccc}
	1&1&0&1\\
	0&1&1&0\end{array}\right].
$$
Then $$\cup_{i \in [4]} \cZ^{(i)} =\{[1 0 0 1],[1 0 0 0],[0 1 1 1],[0 1 0 0],[0 0 1 0],[0 1 1 0]\},$$
so that $\ga(\cI)=2$. In fact  $\langle [0 1 0 0],[0 0 1 0]\rangle\setminus\{0\}\subset \cup_{i \in [4]} \cZ^{(i)}$, and from the $\ga$-bound we have $\INSe{5}=N(2,3)\le\cN(\cI,1)$.
%It results
%$Z^{(1)}=\{(1 0 0 1),(1 0 0 0)\},Z^{(2)}=\{(0 1 1 1),(0 1 0 0)\}$
%$Z^{(3)}=\{(0 0 1 0),(0 1 1 0)\},Z^{(4)}=\{(1 0 0 1),(0 1 1 1)\}$.
It can be checked that $\gk(\cI)=2$, and so from the $\kappa$-bound we have
$5=N(2,3)\geq \cN(\cI,1)$. 
%It is straightforward to check that any matrix representing a binary linear $(\cI,1)$-ECIC must have at least $5$ rows.
%L^1+L^4,L^1,L^3,L^2+L^3,L^1+L^4,L^2+L^3+L^4 all have hamming weight at least 3%
Then
$$
L=\left[\begin{array}{cccc}
1& 0& 1& 0\\
0& 1& 1& 1\\
1& 1& 0& 0\\
1& 1& 1& 0\\
0& 0& 1& 0\end{array}\right]
$$
represents an optimal linear $(\cI,1)$-ECIC.

Let $X=[1 1 1 1]^T$. The sender broadcasts $LX$. Suppose one Hamming error occurs and $U_4$ receives the vector $Y_4=LX+W_{(4)}$, where
$$
W_{(4)}=[0 0 0 1 0].
$$

Then 
$$
Y_4=[0 1 0 0 1].
$$

Let 
$$
M_{(4)}=\left[\begin{array}{cccc}
1& 0& 1& 1\\
0& 0& 0& 1\\
0& 1& 0& 1\\
0& 0& 1& 0\end{array}\right] \quad\mbox{ and } \quad
L'=LM_{(4)}=\left[\begin{array}{cccc}
1& 1& 1& 0\\
0& 1& 1& 0\\
1& 0& 1& 0\\
1& 1& 1& 1\\
0& 1& 0& 1\end{array}\right]
$$

We obtain a parity check matrix of \INSe{$\cC_{(4)}$}, as in \eqref{eq:hi}

$$
H_{(4)}=\left[\begin{array}{ccccc}
0& 0& 0& 1& 1\\
1& 0& 0& 1& 1\\
0& 1& 0& 1& 1\\
0& 0& 1& 1& 1\end{array}\right].
$$

 \INSe{Applying Step I of our decoding algorithm} we obtain
$$
H_{(4)}(Y-L'^{[2]}X'_{[2]})=\left[\begin{array}{ccccc}
0& 0& 0& 1& 1\\
1& 0& 0& 1& 1\\
0& 1& 0& 1& 1\\
0& 0& 1& 1& 1\end{array}\right] 
\left[\begin{array}{c}
						1\\
						1\\
						1\\
						1\\
						1\end{array}\right]=
						\INSe{\left[\begin{array}{c}
						0\\
						1\\
						1\\
						1\end{array}\right]}.
$$

Therefore $\ga_4=0$ and $\gb_4=\INSe{[111]^T}$. Now from Step II, we obtain that the vector $\varepsilon=[0 0 0 0 1]^T$ is a solution of \eqref{eq:eps}  and in Step III  we obtain 
$$
\hat{X}_3=(0-[0 0 0 1 1] \cdot [0 0 0 0 1])/1=1=X_4.
$$

}

\end{example}

\begin{remark}
	\INSe{The above outlined decoding procedure extends that of \cite[Section VII]{Son1} to the ICCSI case, firstly via the use of the matrices $M_{(i)}$, which transform it to an ICSI problem \INSm{for the user $i$}. However, }
	Step III of our algorithm diverges from that one in a different sense: in \cite{Son1} it is required to solve the system $Y_i = L\bar X - \gep$ given $\bar X$ whose coordinates agree with those of $X$ in the side-information of the $i$th user, for any decoding. In our case a pre-computation is performed to determine $h_{(i)}$ as in (\ref{eq:hi}) and $s_{d_i+1}$. We do this by solving the system
	$$
	 h_{(i)}
	 \left[\begin{array}{ccc}
	L'^{d_i+1}& \vline&	L'^{{\overline{[d_i]}}}\end{array}\right]=[1,0,...,0].
	$$ 
%	We can use the Gaussian elimination to solve the system.%, it means that we have a computational complexity equal to $\cO((|\cY_i|+1)N^2)$, where $|\cY_i|+1\le n$.
\INSe{The computational complexities of both algorithms are similar, being dominated by Step II, where a low weight element of a coset of $\cC^{(i)}$ must be found.}
\end{remark}

\INSe{\subsection{Decoding for Rank-Metric Errors}}
In the model presented here, %, given an $N\times n$ matrix $L$ corresponding to an $\fq$-linear index code $\cI=(t,m,n,\cX,f)$. 
we assume that a matrix $Y$ is transmitted and that at any given receiver, a matrix of the form $Y + W$ is received. Therefore the decoding algorithm of the additive matrix channel as described in \cite{sil} may be considered. \INSe{We do not in fact necessarily assume that $L$ represents an $(\cI,\gd)$-ECIC. Instead we add redundancy by embedding the broadcast $LV^{(S)}X$ into a larger matrix with zeroes off the entries assigned to $LV^{(S)}X$. In this scheme, it is assumed that up to $r$ packets are randomly injected into the network, in the form of a matrix of rank $r$. It can be assumed that with high probability, the first $r$ rows of the error matrix are linearly independent.}

Given a $N\times d_S$ matrix $L$ over $\fq$ for an $\cI$-IC each $i$th receiver requires $LV^{(S)}$ and $LV^{(S)}X$ in order to retrieve its requested data
$R_iX$. %If $L$ is known to each receiver then clearly what is required is that $LX$ be received.
%Suppose that $L$ is known to each receiver. 
Employing the method of \cite{sil}, we let 
$$
\INSe{P}=\left(\begin{array}{cc}
0_{v\times v} & 0_{v\times \ell}\\
0_{N\times v} & Q\end{array}\right),
$$
where $Q=LV^{(S)}X \in \fq^{N \times t}$ and $\ell=t$ if $L$ is known to each receiver and $Q=[LV^{(S)}|LV^{(S)}X] \in \fq^{N \times (d_S+t)}$ and $\INSe{\ell= d_S+t}$ if $LV^{(S)}$ is not known to all receivers.
%In the first case the data rate is $\frac{Nt}{(N+v)(t+v)}$ and in the latter it is $\frac{N(n+t)}{(N+v)(n + t+v)}$.

\INSe{}

Given an error matrix $W$ of rank $r\leq v$, we write
$$
W=\left(\begin{array}{cc}
W_{11} & W_{12}\\
W_{21} & W_{22}
\end{array}\right),
$$ 
with $W_{11} \in \fq^{v \times v}$, $W_{21} \in \fq^{N \times v}$, $W_{12} \in \fq^{v \times t}$, $W_{22} \in \fq^{N \times t}$. If $W_{11}$ has rank $r$ then 
$$r= \rk(W_{11}) \leq \rk\left(\begin{array}{c}
W_{11} \\
W_{21} 
\end{array}\right) \leq \rk(W) = r ,$$
so the rows of $W_{21}$ are contained in the row space of $W_{11}$. Therefore, $TW_{11} = W_{21}$ for some $T \in \fq^{N \times v}$. Then  
$$r=\rk (W) = \rk (W_{11})+\rk (TW_{12}-W_{22})= r + \rk (TW_{12}-W_{22}),$$
so we must have $TW_{12}=W_{22}.$
The matrix $T$ can be easily computed, since the submatrices $W_{11},W_{21}$ are known to each receiver. Moreover, since $W_{12}$ is known, the decoder retrieves 
$Q = -TW_{12}+ W_{22}+Q$.

From Lemma \ref{lemdecode}, the matrix $L$ represents an $\fq$-linear $\cI$-IC if and only if for each $i \in [m]$ there exist vectors $\INSe{U} \in \fq^n ,A \in \fq^{d_i}$ and $B\in \fq^N$ such that 
$$R_i = AV^{(i)} - BLV^{(S)} \text{ and } \INSe{U} = AV^{(i)}.$$
Once $LV^{(S)}$ and $LV^{(S)}X$ is known at the $i$th receiver, its requested data $R_iX$ can be computed as follows.
\begin{enumerate}
	\item
	Choose $U \in \cX^{(i)}$. Equivalently, choose $A\in \fq^{d_i}$ and write $\INSe{U}= AV^{(i)}$.
	\item
	Solve $R_i + AV^{(i)} = BLV^{(S)}$ for some $B \in \fq^N$.
	\item
	Compute $R_iX = BY - A \Lambda^{(i)}$.
\end{enumerate} 
In practice, the decoder computes $[S | T]$, the reduced-row echelon form of the matrix
$$
\left[\begin{array}{ll}
V^{(i)} & \Lambda^{(i)}\\
LV^{(S)}       & Y 
\end{array}\right]
$$
and solves for $Z$ in $ZS = R_i$ to retrieve $R_i X = ZT$.
In particular, if $R_i= {\bf e}_j$ for some $j \in [N]$, then $R_i$ already appears as a row of
$P$, and the corresponding row of $Q$ gives the required vector sought.

\INSe{Note that the method of \cite{sil} assumes that the error matrix $W$ has its first $r$ rows linearly independent. This assumption is referred to by the authors
	as {\em error-trapping}.} 
In the event that 
$\rk (W_{11})<\rk \left(\begin{array}{c}
W_{11} \\
W_{21}
\end{array}\right)$, the decoder detects that error-trapping has failed to occur.
If     
$\rk (W_{11})=\rk \left(\begin{array}{c}
W_{11} \\
W_{21}
\end{array}\right) <\rk(W),$
the decoder does not detect that error-trapping has failed, so a decoding failure will occur.
As noted in \cite{sil} this probability is less than
%\begin{gather}\label{eq:prob}
$\frac{2r}{q^{1+v-r}}.$
%\end{gather}            

\begin{remark}
	For the case $t>1$, if $L$ does represent a rank-metric $(\cI,\gd)$-ECIC, and $LV^{(S)}$ is known to each receiver in advance of the transmission, then the sender may broadcast
	$$
	\INSe{P}=\left[\begin{array}{cc}
	0_{v\times v} & 0_{v\times t}\\
	0_{N\times v} & LV^{(S)}X\end{array}\right] $$
	\INSe{While the $i$th client receives the noisy transmission:
	$$P+W=\left[\begin{array}{cc}
	W_{11} & W_{12}\\
	W_{21} & W_{22}+LV^{(S)}X\end{array}\right].
	$$
	If error-trapping has failed and this is detected at the $i$th decoder 
	then from Theorem \ref{th:err}, if $\rk(W_{22}) \leq \gd$ then $R_iX$ is uniquely retrievable from the submatrix $[W_{22}+LV^{(S)}X]$. 
		However, the existence of an {\em efficient} algorithm to compute $R_iX$ given the received matrix $[W_{22}+LV^{(S)}X]$ for an 
		{\em arbitrary} matrix $L$ representing an $(\cI,\gd)$-ECIC is unlikely.} 
\end{remark}

\begin{example}
	\INSm{
		Let $m=4,n=4,q=2$ and let $\cX^{(S)}=\fq^4$.
		Suppose that $\cX^{(i)}$ has dimension $d_i=1$ for each $i \in \{1,...,4\}$.
		Let $\cI$ be the instance defined by 
		user side-information
		$$
		V^{(1)} = \left[\begin{array}{ccccc} 
		0 & 1 & 0 & 0 
		\end{array}\right],
		V^{(2)} = \left[\begin{array}{ccccc} 
		0 & 0 & 1 & 0
		\end{array}\right],
		V^{(3)} = \left[\begin{array}{ccccc} 
		0 & 0 & 0 & 1  
		\end{array}\right],        
		V^{(4)} = \left[\begin{array}{ccccc} 
		1 & 0 & 0 & 0 
		\end{array}\right],
		$$
		and requests 
		$$R_1=[1000],R_2=[0100],R_3=[0010],R_4=[0001].$$
		It can be checked that $\gk(\cI)=3$ and that $$L= \left[\begin{array}{ccccc} 
		1 & 1 & 0 & 0 \\
		0 & 1 & 1 & 0 \\
		0 & 0 & 1 & 1
		\end{array}\right]$$ is a matrix corresponding to an $\cI$-IC.
		Let $v=2$, $t=1$, $X=[1010]$ and suppose that $L$ is known to the receivers. Then the matrix $P$ is given by
		$$
		P=\left(\begin{array}{cc}
		0_{2\times 2} & 0_{2\times 1}\\
		0_{3\times 2} & LX\end{array}\right).
		$$
		Now, suppose a user receives the matrix
		$$
		P+W=\left[\begin{array}{ccc}
		1 & 0 & 1\\
		1 & 1 & 0\\
		0 & 0 & 1\\
		0 & 1 & 0\\
		1 & 0 & 0
		\end{array}\right]; \mbox{ which implies } W=\left[\begin{array}{ccc}
		1 & 0 & 1\\
		1 & 1 & 0\\
		0 & 0 & 0\\
		0 & 1 & 1\\
		1 & 0 & 1
		\end{array}\right]
		$$
		Clearly, the error-trapping has succeeded, since the $2 \times 2$ upper left-hand submatrix of $W$ has rank 2. 
		Then we have 
		$$
		W_{11}=\left[\begin{array}{cc}
		1 & 0 \\
		1 & 1
		\end{array}\right]\mbox{, }W_{21}=\left[\begin{array}{cc}
		0 & 0 \\
		0 & 1\\
		1 & 0
		\end{array}\right],
		$$
		from which we determine 
		$$
		T=\left[\begin{array}{cc}
		0 & 0 \\
		1 & 1 \\
		1 & 0
		\end{array}\right].
		$$
	From the bottom right-hand part of $P+W$ we compute
		$$
		\left[\begin{array}{c}
		1  \\
		0 \\
		0 
		\end{array}\right]- TW_{12}=\left[\begin{array}{c}
		1  \\
		0 \\
		0 
		\end{array}\right]-\left[\begin{array}{cc}
		0 & 0 \\
		1 & 1 \\
		1 & 0
		\end{array}\right]\left[\begin{array}{c}
		1 \\
		0
		\end{array}\right]=\left[\begin{array}{c}
		1  \\
		1 \\
		1 
		\end{array}\right]=LX.
		$$
		At this point the client can decode its demanded message.
	}
\end{example}

\section{Conclusions}

Permitting coded-side information in the index coding problem offers more potential for applications. While the connections to graphs and hypergraphs are no longer apparent as in the classical case, many of the associated combinatorial characterisations have algebraic interpretations. The ICCSI problem in some sense may be viewed as a $q$-analogue of the ICSI problem: the index set of size $d_i$ of the side information of the $i$th user now being replaced by a vector space of dimension $d_i$. This viewpoint means that most bounds on the optimal length of an index code, both for noiseless and noisy channels have analogues in the more general setting of the ICCSI problem. Although it is likely that error-correcting decoding schemes for multicast network coding can be adapted for their index coding equivalents, the design of efficient error-correcting decoding algorithms for index codes remains a challenging problem.

\section{Acknowledgement } 
The authors thank John Sheekey for helpful discussions. \INSe{The authors are grateful to the anonymous referees whose comments led to a great improvement in the presentation of this paper. This work was started following a research visit partially funded by the ESF COST Action IC1104 {\em Random Network Coding and Designs Over $GF(q)$}.}

\end{document}